\pdfoutput=1
\documentclass[10pt,a4paper,reqno]{amsart}
\usepackage{amssymb}
\usepackage{hyperref}
\usepackage{graphicx}
\usepackage{siunitx}
\usepackage{enumerate}
\usepackage{url}
\usepackage{booktabs}
\usepackage{subfigure}
\usepackage{siunitx}
\usepackage{soul}
\usepackage[foot]{amsaddr}
\usepackage{comment}
\usepackage[pdf]{xcolor}
\usepackage{ifthen}
\usepackage[T1]{fontenc}
\usepackage{mathtools}

\usepackage{etoolbox}
\makeatletter
\patchcmd{\@maketitle}
  {\ifx\@empty\@dedicatory}
  {\ifx\@empty\@date \else {\vskip3ex \centering\footnotesize\@date\par\vskip1ex}\fi
   \ifx\@empty\@dedicatory}
  {}{}
\patchcmd{\@adminfootnotes}
  {\ifx\@empty\@date\else \@footnotetext{\@setdate}\fi}
  {}{}{}
\makeatother

\numberwithin{equation}{section}

\theoremstyle{plain}%

\newtheorem{thm}{Theorem}[]

\newtheorem{cor}{Corollary}[]
\newtheorem{alg}{Algorithm}[]

\theoremstyle{definition}%

\newtheorem{ex}{Example}[]

\newcounter{example}

\newcommand{\1}{1 \kern -.41em 1 }

\newcommand{\size}{{\mathrm{size}}\,}

\newcommand{\id}{{\mathrm{id}}}

\newcommand{\conv}{{\mathrm{conv}}\,}

\newcommand{\aff}{{\mathrm{aff}}\,}

\newcommand\R{\mathbb{R}}
\newcommand\Z{\mathbb{Z}}
\newcommand\N{\mathbb{N}}

\renewcommand\P{\mathbb{P}}
\newcommand\Q{\mathbb{Q}}
\newcommand\NP{\mathbb{N}\mathbb{P}}

\newcommand\CA{\mathcal{A}}

\newcommand\CI{\mathcal{I}}
\newcommand\CP{\mathcal{P}}
\newcommand\CR{\mathcal{R}}
\newcommand\CS{\mathcal{S}}
\newcommand\CT{\mathcal{T}}
\newcommand\CF{\mathcal{F}}
\newcommand\CO{\mathcal{O}}
\newcommand\CC{\mathcal{C}}

\newcommand\DS{\displaystyle}

\newboolean{withcomments} 
\setboolean{withcomments}{true} 

\ifthenelse{\boolean{withcomments}}{}{\newcommand{\AComment}[1]{}}
\ifthenelse{\boolean{withcomments}}{\newcommand{\PG}[1]{{\color{blue}[PG: #1]}}}{\newcommand{\PG}[1]{}}

\newenvironment{problem}[1]{
\vspace{.4cm}
\begin{minipage}[c]{0.9\textwidth}
       \noindent
       {#1.}
        \nopagebreak
        \it
        \begin{list}{}
          {\setlength{\labelwidth}{16mm}
           \setlength{\leftmargin}{18mm}
           \topsep2pt
           \itemsep2pt
           
          }
        \nopagebreak
  }
{\end{list}\bigskip
\end{minipage}
}

\begin{document}
\title[Dynamic Discrete Tomography]{Dynamic Discrete Tomography\footnote{NOTICE: this is an author-created, un-copyedited version of an article accepted for publication/published in \textbf{Inverse Problems}. IOP Publishing Ltd is not responsible for any errors or omissions in this version of the manuscript or any version derived from it.  The Version of Record is available online at \url{https://doi.org/10.1088/1361-6420/aaa202}.}}
\author{Andreas Alpers and Peter Gritzmann}
\address{Zentrum Mathematik, Technische Universit\"at M\"unchen, D-85747 Garching bei M\"unchen, Germany}
\email{alpers@ma.tum.de, gritzmann@tum.de}
\thanks{The authors gratefully acknowledge support through the German Research Foundation Grant GR 993/10-2 and the European COST Network MP1207.}

\begin{abstract}
We consider the problem of reconstructing the paths of a set of points over time,
where, at each of a finite set of moments in time the current positions of points in space are 
only accessible through some small number of their X-rays. 
This particular particle tracking problem, with applications, e.g.,  in plasma physics, 
is the basic problem in dynamic discrete tomography.
We introduce and analyze various different algorithmic models. In particular, 
we determine the computational complexity of the problem (and various of
its relatives) and derive algorithms
that can be used in practice. As a byproduct we provide new results on 
constrained variants of min-cost flow and matching problems.
\end{abstract}

\maketitle

\section{Introduction}

In the following, the goal is to determine the paths $\CP_1,\ldots, \CP_n$
of $n$ particles in space over a period of~$t\in \N$ moments in time from 
 images taken by a a fixed number $m$ of cameras. Each spot in each camera image 
is the projection of a particle perpendicular to the
plane of the corresponding camera. The number of particles that project on the same spot can be
detected from the brightness of the image. Hence, in effect, for each moment in time 
we have the information how many particles lie on each line perpendicular to the planes
of the cameras. We will refer to this information as the {\em X-ray images} of the set of particles 
in $m$ directions; details will be given in Section \ref{sect:notation}.

This problem lies at the core of {\em dynamic discrete tomography}, and we will refer to it as
{\em tomographic point tracking} or {\em tomographic particle tracking}. 
As it turns out, the problem comprises two different but coupled basic underlying 
tasks, the reconstruction of a finite set of points from few of their X-ray images 
({\em discrete tomography}) and the identification of the points over time
({\em tracking}). For surveys on various aspects of discrete tomography
see~\cite{gardner-gritzmann-99, gritzmann97, gritzmann-devries-2001, herman-kuba-99, herman-kuba-07}. 
Tracking is known in the literature as {\em data association} or {\em object tracking}
and, for more specific applications, as {\em multi-target tracking},  
{\em multisensor surveillance} or, in physics, as {\em particle tracking}; 
see \cite{objecttracking, surveillance, multitracking7, multitracking2, multitracking1, elementparticletracking, 
schnoerr1, multitracking4} and, e.g., \cite{a2010, assignmentbook09, multitracking6} for surveys.

Tomographic particle tracking was already considered in 
\cite{DPS17, Elsinga2006, kitzhofer-10, batenburg10, Williams2011} and \cite{agms-15}. 
In fact the \emph{linear programming} based heuristic introduced in  \cite{agms-15} was successfully 
applied to  determine 3D-slip velocities of a gliding arc discharge in~\cite{glidingarc-15}. 
Another \emph{linear programming} approach, different from~\cite{agms-15, glidingarc-15}, was recently proposed in~\cite{DPS17}.

In the present paper we study the problem from a mathematical and
algorithmic point of view with a special focus on the interplay
between discrete tomography and tracking. Therefore
we will distinguish the cases that for none, some or all of the
$t$ moments $\tau_1,\dots,\tau_t$ in time, a solution of the discrete tomography task at time $\tau_1,\dots,\tau_t$ 
is explicitly available (and is then considered {\em the} correct solution  
regardless whether it is uniquely determined by its X-rays). 
For $d\ge 3$, two (affine) lines in $\R^d$ in general position are disjoint.
Hence, generically, X-ray lines meet only in points of the underlying set.
Therefore even the latter situation is of considerable practical relevance. 
We will refer to it as the {\em positionally determined} case while, 
in the more general situation, we will speak of the {\em (partially)} or {\em (totally) tomographic} case of point tracking.

In the following, we discuss several models for dynamic discrete tomography, give various algorithms and 
complementing $\NP$-hardness results.
In particular, we show for the positionally determined case that the tracking problem can be solved in polynomial time 
if it exhibits a certain Markov-type property (which, effectively, allows only dependencies between 
any two consecutive time steps); see Theorem~\ref{thm:ILP-posdet}. 
The partially tomographic case, however, is $\NP$-hard even for~$t=2$; 
see Theorem~\ref{thm:tomo-matching}. 
Complementing this result, we consider the tomographic tracking problem for two directions where a 
so-called \emph{displacement field}  is assumed to be given (the displacement field uniquely determines the particle's 
next position). Again, this problem turns out to be $\NP$-hard already for~$t=2$ and certain realistic 
classes of displacement fields; see Theorem~\ref{thm:tomo-circular}. 

We then turn to the {\em rolling horizon} approach introduced in~\cite{agms-15} which proceeds successively 
from step to step. After giving a short account of this method in Section \ref{subsec-rolling-horizion} we will
show that, while being quite successful in practice, it is not guaranteed to always yield the correct solution; see Example~\ref{ex:grid1}.
Then, in~Section~\ref{sect:nonmarkov} we study the issue of of how to  incorporate additional prior 
knowledge of particle history into the models. Under rather general assumptions we show that already in the 
positionally determined case the tracking problems becomes $\NP$-hard for~$t>2;$ 
 see Theorem~\ref{thm:matching}. In particular, we show that even if the particles are known to move 
along straight lines, this prior knowledge cannot efficiently be exploited algorithmically 
(unless~$\P=\NP$); see Theorem~\ref{thm:special-weights} and Corollary~\ref{cor:straightline} 
and~\ref{cor:minmax}.

We proceed by introducing three algorithms that can be viewed as rather general paradigms of heuristics 
that involve prior knowledge about the movement of the particles and which can be used in the tomographic case. 

We then discuss \emph{combinatorial models}. In these models the positions of the particles in the next time step 
are assumed to be known approximatively in the sense that the candidate positions are confined to certain 
\emph{windows}, which are finite subsets of positions. Again, under rather general conditions, we show 
$\NP$-hardness of the respective tomographic tracking tasks; see Theorem~\ref{thm:comb2}. 
However, we also identify polynomial-time solvable special cases of practical 
relevance; see Theorem~\ref{thm:orth} and~\ref{thm:comb5}. 

The paper is organized as follows: Section~\ref{sect:notation} will introduce the relevant notation, provide 
some basic background, discuss various modeling issues including the question of how to rigorously 
incorporate a notion of `physically most likely' solution, and state our main results.

Then main focus of the present paper lies on optimization models. They are studied in quite some
detail in Section~\ref{sect:optmodels}. In particular, we consider Markov-type integer programming and rolling horizon 
models and discuss the issue of utilizing the particle history from a structural algorithmic point of view. Also, we determine the 
computational complexity of the problems. We prove various $\NP$-hardness results, both, for the partially tomographic and the totally tomographic case, 
in the latter, even when the displacement field is known. Then we give various heuristic algorithms (or, to be more precise, classes of 
algorithmic paradigms) which are designed with a view towards balancing the running time and the achieved quality of
the solutions. In the prevailing presence of $\NP$-hardness such a balancing is mandatory for a successful use 
in practical applications (of realistic size).

Section~\ref{sect:combinatorial} deals with a different way of encoding a priori information about
the movement of the particles. The studied combinatorial models utilize such information in terms of bounds 
on the number of particles in certain windows. We provide $\NP$-hardness results but also identify
polynomially solvable classes of instances. 
The paper concludes  with some final remarks in Section~\ref{sect:conclusion}.

\section{Notation, basics and main results}\label{sect:notation}

\subsection{Discrete Tomography}
Let $\mathbb{R},$ $\mathbb{Q},$ $\mathbb{Z},$ and $\mathbb{N}$ denote the set of reals, rational numbers, integers,  and natural numbers, respectively. For $q\in \N$, let $[q]=\{1,\ldots,q\}$ and $[q]_{0}=[q]\cup \{0\}$. With $\1$ we denote the all-ones vector. 
Further, let~$d,m \in\N$, $d\ge 2$, and set  
$$
\CF^d=\{F: F\subseteq \Q^d \, \land \, \textrm{$F$ is finite}\}.
$$ 
The point sets $F\in \CF^d$ model the sets of particles (in a static
environment). Let us point out that, mathematically, the restriction to $\Q^d$ 
is merely a tribute to the model of computation
that we are going to apply. In fact, we will use the {\em binary Turing machine model}
where each element is encoded in binary and hence its size is the number of
binary digits needed for this representation; see e.g. \cite{gareyjohnson}. 
(The binary size of a positive integer~$q$ is 
therefore essentially $\log(q)$.) Note that this restriction is also in accordance with 
standard practical measurements which do not allow infinite precision. 

With $\CS^d$ we denote the set of all {\em lattice lines}, i.e., the set of $1$-dimensional linear 
subspaces of $\R^d$ that are spanned by vectors from $\Z^d$. 
(Note that as $\CS^d$ coincides with the set of lines spanned by rational
vectors, there is no further restriction here.)
The lines in $\CS^d$ are (in a slight abuse of language) often referred to as directions.
In fact, in an experimental environment, these lines are the `viewing directions' 
under which the cameras `see' the particles i.e., they are perpendicular to the 
image planes of the cameras.

For $S\in \CS^d,$ we set $\CA(S)=\{v+S: v\in \Q^d\}.$
Then, for $F\in \CF^d$ and $S\in \CS^d,$ the {\em (discrete $1$-dimensional) 
X-ray of $F$ parallel to $S$} is the function
$$
X_S F:{\CA}(S)\rightarrow \N_0=\N\cup\{0\}
$$ 
defined by 
$$
X_SF(T) = |F\cap T|=\sum_{x\in T}\chi_F(x)
$$ 
for each $T\in {\CA}(S),$ where, as usual, $|F\cap T|$ denotes the cardinality of $F\cap T$ 
and $\chi_F$ is the characteristic function of $F$.

Let us point out that this definition is the basis of the paramount {\em grid model} 
of discrete tomography (see~\cite{gardner-gritzmann-99, gritzmann97, 
gritzmann-devries-2001, herman-kuba-99, herman-kuba-07} which we will use throughout
the paper. Of course, the term X-ray is meant generically here, i.e., does not necessarily 
refer to a specific imaging technique. However, the main motivation underlying the present paper 
is that of tracking physical particles which are `visible' only through the camera image of their 
projections. 

Two sets $F_1,F_2 \in \CF^d$ are called {\em tomographically equivalent} with respect to $S_1,\dots,S_m \in \CS^d$ if \[X_{S_i} F_1= X_{S_i} F_2,\]  for all  $i\in [m].$

Given $m$ different lines $S_1, \dots, S_m \in \CS^d $, the X-ray data is given in terms
of functions 
$$
f_i:{\mathcal A}(S_i)\rightarrow \N_0 \qquad (i \in [m]),
$$ 
with finite support $\mathcal{T}_i \subseteq \mathcal{A}(S_i)$,
represented by appropriately chosen data structures; see~\cite{ggp-99}.

Suppose, we are given an instance $(f_1,\ldots,f_m)$ of measurements of an otherwise unknown set
$F\in \CF^d.$ Then, of course, $|F|=\|f_1\|_{(1)}=\ldots= \|f_m\|_{(1)}$,
where $\|\cdot\|_{(1)}$ is the usual 
$1$-norm. From the data functions we can infer that $F$ must be contained in 
the {\em grid}
$$
G= \bigcap_{i=1}^m \bigcup_{T\in \CT_i} T
$$
of candidate points.
Note that, in general, even when the given instance has a unique solution 
$F,$ the grid $G$ can be a proper superset of $F.$  

\subsection{Dynamic Discrete Tomography}
In tomographic point tracking, we consider $t$ consecutive moments in 
time. Hence, in particular, we are interested in the sets $F^{(\tau)}$ of particles 
at each moment in time $\tau\in [t]$. 
Since, in general, these sets are accessible only through their X-rays taken from a given finite set of directions, they represent $t$ {\em uncoupled} instances of 
discrete tomography.

Of course, we also need to track the particles over time.
Let $P=\{p_1,\ldots,p_n\}$ denote the (abstract) set of $n$ particles. 
Then, for each $\tau\in [t]$,  we are actually interested in a
one-to-one mapping $\pi^{(\tau)}: P \rightarrow F^{(\tau)}$ that identifies the points
of $F^{(\tau)}$ with the particles. Hence, 
the particle tracks are given by $\CP_i=(\pi^{(1)}(p_i)\ldots, \pi^{(t)}(p_i)),$ $i\in [n].$ 
We will refer to this identification as {\em coupling}. 

Note that the coupling between two consecutive moments in
time bears the character of a {\em matching} in a bipartite graph. 
In tomographic point tracking, however, the tasks of discrete tomography 
and matching strongly interact. 

As it is well-known that for $m\ge 3$ already the static reconstruction 
problem of discrete tomography is $\NP$-hard \cite{ggp-99} (see also 
\cite{bdgv-08}), highly instable \cite{alpers-d03, agt-01}, and since in practical experiments space for installing
cameras is typically strictly limited (see, e.g, \cite{agms-15}) we will mainly 
focus on the case $m=2$. 
In real-world particle tracking this corresponds to images taken from two cameras 
which may be regarded as two planes in~$\R^3$ orthogonal to 
two different directions $S_1$, $S_2$, respectively.

For $m=2,$ a set $F^{(\tau)}$ satisfying the X-ray constraints with respect to 
the two given directions~$S_1$ and~$S_2$ can be efficiently determined for each 
$\tau \in [t]$; see \cite{gale57,ryser57}. Also, it can be checked efficiently, 
whether each $F^{(\tau)}$ is unique. (For related stability results see~\cite{alpers-brunetti-07,van2009stability}.)
Since a set is only very rarely uniquely determined by just two of its X-rays,  
the reconstruction will typically rely on additional physical knowledge.
(For various uniqueness 
results the reader is referred to \cite{glw11} and the literature quoted therein.) 

But even if we know for each $\tau\in [t]$ the correct set $F^{(\tau)},$ 
we still have to identify the paths of all individual particles over time. 
This turns the~$t$ otherwise uncoupled systems into a single coupled system. 
Here we will in particular discuss the question of how to utilize additional
information via the coupling, which thus might reduce the ambiguity of the uncoupled
tomographic tasks.

But how can the goal of determining a `physically most likely' solution
be modeled?  
Let $\Pi(n,t)$ denote the set of all $n$-tuples $(\CP_1,\ldots,\CP_n)$
of potential particle tracks over~$t$ moments of time for fixed point sets 
$F^{(\tau)}$. If the given tomographic particle tracking data is feasible
there may, of course, exist many different but tomographically 
equivalent solutions for each $\tau$. In any case, $|\Pi(n,t)|$ is then 
still a lower bound for the number of different potential solutions.
Now note that $|\Pi(n,t)|=(n!)^t$, and even if the order of the individual 
paths is irrelevant,  this number reduces 
only to $(n!)^{t-1},$ which is still exponential in $n$ and $t$. 

Of course, this means that any enumerative algorithm will fail in practice even
for quite moderate numbers of particles and time steps.  But even worse,
we cannot encode as input the `physical costs'~$c(\CP_1,\ldots,\CP_n)$ of all $(\CP_1,\ldots,\CP_n)$
explicitly. Even if we assume that all sets $F^{(\tau)}$ are given and
such costs can be computed as a simple function 
of the costs $w(\CP_i)$ of the individual paths, e.g., as
\[\sum_{i=1}^n w(\CP_i)\qquad \mbox{or} \qquad \max_{i\in [n]} \,w(\CP_i),\]
the number of different costs $w(\CP)$ of a particle path (which still need to be
available to determine the cost of a solution$(\CP_1,\ldots,\CP_n)$) reduces only to $n^t.$ 

In the following we will therefore in general refrain from assuming that such values 
are {\em explicit} parts of the input. We will instead assume that
an algorithm is available which computes for any solution~$(\CP_1,\ldots,\CP_n)$ 
its cost $c(\CP_1,\ldots,\CP_n)$ in time that is polynomial in all the other input data.
In some situations such an algorithm will be based on additional specific data 
which reflect how appropriate certain (local) choices are and which will then be 
regarded as be part of the input, too. 
Such an algorithm $\CO$ will be called an {\em objective function oracle}.
In the special case that $\CO$ provides the values~$w(\CP_i)$ and 
$c(\CP_1,\ldots,\CP_n)=\sum_{i=1}^n w(\CP_i),$ the oracle will be referred to as
{\em path value oracle}. 

For all practical algorithms, such an oracle will of course be specified
explicitly. (It will typically be based on weights on grid points or
pairs of grid points which augment the input of our problem.) 

\subsection{Basic algorithmic problems}
Algorithmically, our general problem of tomographic particle tracking 
for $m$ given different directions~$S_1,\ldots ,S_m \in \CS^d,$
and based on an objective function oracle $\CO$,
can now be formulated as follows.

\begin{problem}{{\sc TomTrac}${(\CO; S_1,\ldots ,S_m)}$}
\item[Instance] $t\in \N,$ X-ray data functions 
$f^{(\tau)}_1, \ldots, f^{(\tau)}_m$ for $\tau\in [t]$ with 
$\|f^{(\tau)}_1\|_{(1)}= \ldots = \|f^{(\tau)}_m\|_{(1)}=n$.
\item[Task] Decide whether, for each $\tau\in [t]$, there exists a set 
$F^{(\tau)}\in \CF^d$ such that $X_{S_i}F^{(\tau)}=f^{(\tau)}$ for all $i\in [m]$.
If so, find particle tracks $\CP_1,\ldots,\CP_n$ of
minimal cost for $\CO$ among all couplings of all 
tomographic solutions $F^{(1)}, \ldots,F^{(t)}$.
\end{problem}

Let us add a remark concerning the assumption, that the $1$-norms of
the X-ray data functions in {\sc TomTrac}${(\CO; S_1,\ldots ,S_m)}$ 
are known and assumed to coincide. First, note that the latter condition is quite 
natural if all particles are detected at each moment in time. 
Also, both conditions can be checked in polynomial time.
Since in practice, cameras have a finite field of view it may, however, happen
in an experimental setting that several previously detected particles may 
move out of the field of view of the camera. Also particles may appear 
(or reappear) during the measurements. In this situation we might at least 
aim at obtaining partial results by reconstructing partial tracks or by including dummy nodes, which account for invisible particles (for a discussion, see, e.g.,~\cite{missingdata}). 
In the following, our prime focus will lie, however, on the case that all 
$n$ particles are recorded at each moment in time $\tau\in [t].$

In the following we will distinguish the cases that for none, some or all 
$\tau\in [t]$, the correct solution~$F^{(\tau)}$ is explicitly given. 
The former will be referred to as the {\em (partially)} or {\em (totally) 
tomographic} case while we speak of the latter as  
{\em positionally determined}. In the positionally determined case the 
problem {\sc TomTrac}${(\CO; S_1,\ldots ,S_m)}$ reduces to the following problem.

\begin{problem}{{\sc Trac}$(\CO;d)$}
\item[Instance] $t\in \N$, sets $F^{(1)},\ldots,F^{(t)}\in \CF^d$ with
$|F^{(1)}|= \ldots = |F^{(t)}|=n$.
\item[Task] Find particle tracks $\CP_1,\ldots,\CP_n$ of
minimal cost for $\CO$ among all couplings of 
the sets $F^{(1)}, \ldots,F^{(t)}$.
\end{problem}

Note that each instance of {\sc Trac}$(\CO;d)$ can be viewed 
as a $t$-dimensional assignment problem. For a comprehensive survey on assignment
problems see \cite{assignmentbook09} and the references quoted therein. 

As pointed out before, for $d\ge 3$, the positionally determined case
is the generic situation even for~$m=2$ as two lines parallel to $S_1$ and $S_2$ 
do only meet in points of $F^{(\tau)}$. 
However, for~$d=2$, any two non-parallel lines in the plane meet in a point, 
whence forming a candidate grid $G^{(\tau)}$ of size $O(n^2)$. 
(Note that, generally, no fixed number $m$ of X-ray images suffices 
for unique determination.)
Also, even  in~$\R^d$ with~$d\ge 3,$ there are particle constellations 
for which the candidate grids $G^{(\tau)}$ do not consist of just~$n$ points. 
And, of course, if the resolution of the images is low, such constellations 
become relevant even in practice. 

In some applications, there is strong prior knowledge about the possible couplings available. 
Most notably is the case where a \emph{displacement field} is prescribed. 
Formally a \emph{particle displacement field} (for the time step $\tau\to\tau+1$) is a 
pair $(p,\varphi^{(\tau)}(p)-p)$ with~$\varphi^{(\tau)}:\Q^d\to\Q^d$ denoting 
a map where~$\varphi^{(\tau)}(p)-p$ denotes the displacement vector for a particle 
located at position $p$ at time~$\tau.$ In particular, if $\pi^{(\tau)}(p_i)$ denotes 
the position of the $i$-th particle at time~$\tau,$ then 
$\pi^{(\tau+1)}(p_i)=\varphi^{(\tau)}(\pi^{(\tau)}(p_i)),$ for  $i\in[n]$ and $\tau\in[t-1].$ 
Now, the tomographic tasks for the time step $\tau\to \tau+1$ consists of determining 
the positions of the particles at time $\tau$ such that the solution at time~$\tau+1$ 
determined by the displacement field matches the tomographic data for~$\tau+1.$ 
We say that $F^{(\tau)}$ and $F^{(\tau+1)}$ are \emph{$\varphi^{(\tau)}$-compatible} if 
both sets satisfy the corresponding tomographic constraints 
and $\varphi^{(\tau)}(F^{(\tau)})=F^{(\tau+1)}.$ 

For the following problem we suppose that a for any $\tau\in[t-1]$ a
displacement field $(p,\varphi^{(\tau)}(p)-p),$ $p\in\Q^d,$ is known.  
While for our main $\NP$-hardness result, the displacement field is given
explicitly, all that our model of computation actually requires is that 
for any $\tau\in[t-1]$ and $p\in G^{(\tau)},$ either a point $q\in G^{(\tau+1)}$ is
provided such that $\varphi^{(\tau)}(p)=q$ or it is reported that no such point exists.
(It is not difficult to model this formally again as some {\em displacement field oracle}).
With this understanding, the tomographic point tracking problem with given 
displacement field by the oracle~$\CO$ is as follows.

\begin{problem}{{\sc TomDisplaceTrac}${(\CO; S_1,\ldots ,S_m)}$}
\item[Instance] $t\in \N,$ X-ray data functions 
$f^{(\tau)}_1, \ldots, f^{(\tau)}_m$ for $\tau\in [t]$ with 
$\|f^{(\tau)}_1\|_{(1)}= \ldots = \|f^{(\tau)}_m\|_{(1)}=n$; displacement field $\varphi^{(\tau)}$ for $\tau\in [t-1]$.
\item[Task] Find sets $F^{(1)},\dots,F^{(t)}\in\CF^d$ such that $X_{S_i}F^{(\tau)}=f^{(\tau)}$ for all $i\in [m],$ and such that 
$F^{(\tau)}$ and $F^{(\tau+1)},$  are $\varphi^{(\tau)}$-compatible for all $\tau\in[t-1],$ 
or decide that no such sets exist. 
\end{problem}

\subsection{Main results}
Our main results are as follows.
First, we consider~{\sc TomTrac}${(\CO; S_1,\ldots ,S_m)}$ for \emph{Markov-type} models, i.e., for  
objective function oracles $\CO$ that can be given explicitly as the sum of all costs of assigning points between consecutive moments in time. We show that the corresponding version of {\sc TomTrac}${(\CO; S_1,\dots,S_m)}$ for the positionally determined case, called {\sc Trac}$(d),$ can be solved in polynomial time (Theorem~\ref{thm:ILP-posdet}).

For the tomographic case, we show that 
even if all instances are restricted to~$t=2$, and the solution~$F^{(1)}$ is given explicitly, {\sc TomTrac}${(S_1,S_2)}$ is $\NP$-hard (Theorem~\ref{thm:tomo-matching}). (Also the corresponding uniqueness problem is $\NP$-hard and the counting problem 
is $\#\P$-hard.) This result is complemented by Theorem~\ref{thm:tomo-circular}, which shows that {\sc TomDisplaceTrac}${(\CO; S_1,S_2)}$ is $\NP$-hard for $t=2$ and certain conditions imposed on the specified displacement field.

Theorem~\ref{thm:matching} shows that the problem {\sc Trac}$(\CO;d)$ is $\NP$-hard, even if all instances are restricted to a fixed $t\ge 3,$ and $\CO$ is a path value oracle. Also in the case of straight line movement and fixed $d\geq 2,$ $t\geq 3$ it is $\mathbb{N}\mathbb{P}$-complete problem to decide whether a solution of {\sc Trac}$(\CO;d)$ exists where all particles move along straight lines (Corollary~\ref{cor:straightline}).

We conclude Section~\ref{sect:optmodels} by introducing and discussing 
algorithmic paradigms that allow to include prior knowledge of the possible particle tracks. 

Section~\ref{sect:combinatorial} then studies combinatorial models. In particular, we discuss {\sc Tomography under Window Constraints}, which, under rather mild conditions, turns out to be $\NP$-hard (Theorem~\ref{thm:comb2}). On the other hand, we show in Theorem~\ref{thm:orth} that  it is in $\mathbb{P}$ if all windows are disjoint horizontal and vertical windows of width~1. Another polynomial-time solvable variant of {\sc Tomography under Window Constraints}, which arises in super-resolution imaging applications, is discussed in Theorem~\ref{thm:comb5}. 

\section{Optimization Models}\label{sect:optmodels}

In the following we assume that we have X-ray measurements from two 
different directions $S_1,S_2\in \CS^d$ at $t$ moments in time.
As before, $G^{(\tau)}$ will denote the corresponding grid at $\tau \in [t]$. 

While the grids $G^{(\tau)}$ are subsets of $\R^d$, we need to distinguish a point 
$g^{(\tau_1)}\in G^{(\tau_1)}$ from a point $g^{(\tau_2)}\in G^{(\tau_2)}$ 
when $\tau_1\ne \tau_2$ even
if, physically, both points occupy exactly the same position in $\R^d$. 
So, formally,~$g^{(\tau)}$ must be regarded as a point $(g,\tau) \in \R^d \times [t]$.
In the following, however, we will not always `verbally' distinguish between the interpretations
$G^{(\tau)}\subseteq \R^d$ and $G^{(\tau)}\subseteq \R^d\times [t]$ if there is no 
risk of confusion.

In the positionally determined case where the correct sets
$F^{(\tau)}$ are known, we can, of course, directly work with $F^{(\tau)}.$ 
Hence, whenever we speak of the positionally determined case, it is in the following
tacitly assumed that $G^{(\tau)}=F^{(\tau)}.$

\subsection{A Markov-type integer programming model}\label{sect:Markov}

In this section we will consider the problem~{\sc TomTrac}${(\CO; S_1,\ldots ,S_m)}$ for  
objective function oracles $\CO$ that can be given explicitly since their values 
are just the sums of all costs of assigning points between consecutive moments in time.
We call these models {\em Markov-type}  since the objective function reflects 
only dependencies that occur between neighboring layers. In Section
\ref{sect:nonmarkov} we will discuss more general models.

Let us now present an integer linear programming ({\sc Ilp}) model for this
problem. 
In order to make the notation as transparent as possible we number the
points of each grid $G^{(\tau)}$ as $g_i^{(\tau)}$ for $i\in  I^{(\tau)}=[\,|G^{(\tau)}|\,]$. 
Also, we set $e_{i,j}^{(\tau)}=(g_i^{(\tau)}, g_j^{(\tau+1)})$, and refer to such 
pairs as {\em tracking edges}.
The objective function will then depend only on weights $\omega_{i,j}^{(\tau)},$
associated with the tracking edges, which measure the
cost of assigning a point $g_i^{(\tau)}\in G^{(\tau)}$ to a point 
$g_j^{(\tau+1)}\in G^{(\tau+1)}$. 
Explicitly expressed, we have
\[
c(\CP_1,\ldots,\CP_n)= \sum_{k=1}^n \sum_{\tau\in [t]} \sum_{e_{i,j}^{(\tau)}\in \CP_k}
\omega_{i,j}^{(\tau)}.
\]
Setting $w(e)= \omega_{i,j}^{(\tau)}$ for $e= e_{i,j}^{(\tau)}$, we will 
refer to $w$ as the {\em weight function} realizing the Markov-type 
objective function oracle $\CO$. Here, of course, the full list of
all weights $w(e)$ is part of the input of our problem.
Note that the number of entries in this list is bounded from above by $(t-1)n^4$
even in the totally tomographic case. Hence there are at most
polynomially many rational numbers to be presented.
With this specification we refer to our respective problems as {\sc TomTrac}${(S_1,\ldots ,S_m)}$ and 
{\sc Trac}$(d).$

Note that the weights $w(e)$ can be used  to model additional `local' physical knowledge 
such as information about the potential ranges of velocities or directions of 
the particles.

Now, we introduce two sets of $0$-$1$-variables associated with the grid points 
and the tracking edges, respectively.
More precisely, for $\tau \in [t]$ and $i\in  I^{(\tau)},$ the variable $\xi^{(\tau)}_i$ 
corresponds to $g_{i}^{(\tau)}$ while the variable $\eta_{i,j}^{(\tau)}$ 
is associated with $e_{i,j}^{(\tau)}.$ 

The {\em tomographic variables} $\xi^{(\tau)}_i$ will be used to describe 
the tomographic constraints; $\xi^{(\tau)}_i=1$ signifies that the grid point 
$g_{i}^{(\tau)}$ is actually present in the computed solution $F^{(\tau)}.$ 
The {\em tracking variable} $\eta_{i,j}^{(\tau)}$ has value $1$ if, and only if, the particle 
that, at time $\tau$ is located at $g_i^{(\tau)},$ moves to $g_j^{(\tau+1)}.$ 
For a compact notation, we set 
$x^{(\tau)}=(\xi_1^{(\tau)},\dots,\xi_{|G^{(\tau)}|}^{(\tau)})^T,$
write $A^{(\tau)}$ for the corresponding coefficient matrix,  and 
encode the X-ray information in the `right-hand side vector' $b^{(\tau)}.$

With this notation we can in principle solve the problem as an  {\sc Ilp}.

\begin{alg}[{\sc Tomographic Tracking-Ilp}]\label{alg:ILP}
Solve the following integer linear program:
\begin{equation} \label{eq:model1}
\begin{array}{lrclcl}
&\multicolumn{4}{c}{\DS\min \,\, \sum_{\tau\in [t-1]}\sum_{i\in I^{(\tau)}}\sum_{j\in I^{(\tau+1)}} \omega_{i,j}^{(\tau)}\eta_{i,j}^{(\tau)}}&\\[.6cm]
\textnormal{s.\,t. } &A^{(\tau)}x^{(\tau)}      & = & b^{(\tau)}, &\quad &(\tau\in[t]),\\[.2cm]
&\DS \xi_i^{(\tau)}-\sum_{j\in I^{(\tau+1)}}\eta_{i,j}^{(\tau)}&=&0, &&(i\in I^{(\tau)},\:\tau\in[t-1]),\\[.2cm]
&\DS \xi_j^{(\tau+1)}-\sum_{i\in I^{(\tau)}}\eta_{i,j}^{(\tau)}&=&0, &&(j\in I^{(\tau+1)},\:\tau\in[t-1]),\\[.2cm]
&\eta_{i,j}^{(\tau)}&\in& \{0,1\}, &&(i\in I^{(\tau)},\:j\in I^{(\tau+1)},\:\tau\in[t-1]),\\[.2cm]
&         x^{(\tau)}    & \in &\{0,1\}^{|G^{(\tau)}|},&&(\tau\in[t]).\\
\end{array}
\end{equation} 
\end{alg}

Note that the coefficient matrices $A^{(\tau)}$ are totally unimodular (see, e.g.,~\cite{aharoni97,gritzmann97}) and the first set of equality constraints contains only the tomographic 
variables.
The second and third set of equalities in~\eqref{eq:model1} couple
the tomographic and the tracking variables. As the tomographic variables are 
$0$-$1$ they guarantee two properties: (1)~A  tracking
edge can only connect grid points that are present in the tomographic solution.
The second set of constraints corresponds to the edges `leaving' 
time $\tau$ while the third set corresponds to those `entering' time
$\tau+1$. (2)~From each point $g_i^{(\tau)}$ in a considered solution, i.e.,
when~$\xi_i^{(\tau)}=1$, exactly one `leaving' tracking edge is selected.
Similarly, when $\xi_i^{(\tau+1)}=1$, there must be exactly one `entering' 
tracking edge.

Note that  by adding the two constraints
\[
\xi_i^{(\tau+1)}-\sum_{j\in I^{(\tau+1)}}\eta_{j,l}^{(\tau+1)}=0 \qquad \mbox{and}
\quad 
-\xi_i^{(\tau+1)}+\sum_{j\in I^{(\tau)}}\eta_{j,i}^{(\tau)}=0,
\] for $i\in I^{(\tau)}$ and $\tau \in \{2,\dots,t-1\},$
we derive the set of constraints
\begin{equation} \label{eq:match}
\sum_{j\in I^{(\tau)}}\eta_{j,i}^{(\tau)}-\sum_{j\in I^{(\tau+1)}}\eta_{i,j}^{(\tau+1)}=0,
\end{equation}
which express that, if a particle enters $g_i^{(\tau)}$ it must exit it
again, and vice versa. Hence, the equations~\eqref{eq:match} can be viewed 
as {\em path} or {\em flow constraints}, and the corresponding coefficient matrix 
is again totally unimodular. When we replace appropriate conditions in
 \eqref{eq:model1} by these flow constraints, we obtain an equivalent but different 
{\sc Ilp}-formulation which decomposes into two totally unimodular parts and 
a reduced set of coupling constraints which, however, still destroy the 
total unimodularity of the whole system. 
As it will turn out in the next subsection, this is not merely a nuisance but a 
severe obstacle for efficient algorithms.

We prefer the formulation \eqref{eq:model1} because it reveals the 
uncoupled structure in the positionally determined case
and leads to the following result for  {\sc Trac}$(d)$. (Recall that now $\CO$
is a Markov-type objective function oracle which is realized by a rational 
weight function $w$.)

\begin{thm}\label{thm:ILP-posdet}
{\sc Trac}$(d)$ decomposes into uncoupled minimum weight perfect 
bipartite matching problems and can hence be solved in polynomial time.
\end{thm}

\begin{proof} 
Under the assumptions of this theorem,  \eqref{eq:model1} reads

\begin{equation} \label{eq:ILP-posdet}
\begin{array}{lrclcl}
&\multicolumn{4}{c}{\DS\min \,\, \sum_{\tau\in [t-1]}\sum_{i\in I^{(\tau)}}\sum_{j\in I^{(\tau+1)}} \omega_{i,j}^{(\tau)}\eta_{i,j}^{(\tau)}}&\\[.6cm]
\textnormal{s.\,t. }&\DS \sum_{j\in I^{(\tau+1)}}\eta_{i,j}^{(\tau)}&=&1, &&(i\in I^{(\tau)},\:\tau\in[t-1]),\\[.2cm]
&\DS \sum_{i\in I^{(\tau)}}\eta_{i,j}^{(\tau)}&=&1, &&(j\in I^{(\tau+1)},\:\tau\in[t-1]),\\[.2cm]
&\eta_{i,j}^{(\tau)}&\in& \{0,1\}, &&(i\in I^{(\tau)},\:j\in I^{(\tau+1)},\:\tau\in[t-1]).\\
\end{array}
\end{equation} 

Note that, here, in the positionally determined case, we have
$|I^{(\tau)}|=|G^{(\tau)}|=|F^{(\tau)}|=n$.
Hence we need to solve the $t-1$ independent minimum weight perfect bipartite matching tasks

\begin{equation} \label{eq:ILP-posdet2}
\begin{array}{lrclcl}
&\multicolumn{4}{c}{\DS\min \,\, \sum_{i,j \in [n]} \omega_{i,j}^{(\tau)}\eta_{i,j}^{(\tau)}}&\\[.6cm]
\textnormal{s.\,t. }&\DS \sum_{j\in [n]}\eta_{i,j}^{(\tau)}&=&1, &&(i\in [n]),\\[.2cm]
&\DS \sum_{i\in [n]}\eta_{i,j}^{(\tau)}&=&1, &&(j\in [n]),\\[.2cm]
&\eta_{i,j}^{(\tau)}&\ge& 0, &&(i,j\in [n]).\\
\end{array}
\end{equation} 

In \eqref{eq:ILP-posdet2} the condition $\eta_{i,j}^{(\tau)}\in \{0,1\}$ has been replaced by the
non-negativity constraints $\eta_{i,j}^{(\tau)}\ge 0$. Since, by the other constraints, 
$\eta_{i,j}^{(\tau)}\le 1$, we have just switched to the {\em {\sc Lp}-relaxation}. The feasible points of this {\sc Lp}-relaxation form a polytope with integral vertices \cite{birkhoff} (see also~\cite{gritzmann97}), hence each of these problems can be solved in polynomial time; see, e.g., \cite[Sect.~16~and~19]{schrijver-86}.
\end{proof}


Let us close this section by turning again to the ILP~\eqref{eq:model1}. Clearly, the ILP contains~$O(tn^2)$ constraints and~$O(tn^4)$ binary variables, which, in the positionally determined case, reduces to $O(tn)$ constraints and $O(tn^2)$ binary variables. 
As the computation time depends strongly on the structure of the problem instance, it is not clear which computation times will occur for a given practical instance. For the positionally determined case, however, the computational study in~\cite{steger} shows that random instances with~$n$ up to~$20,000$ and $t=2$ can be solved with state-of-the-art algorithms in reasonable time.  Note, however, that without resorting to sparsity and advanced optimization techniques it seems hopeless to solve instances in the tomographic case with $n=t=30$ as the coefficient matrix can involve~$54,000\times 23,517,000$ entries (amounting to more than 1~Terabyte of storage space). Of course, it is always possible to resort to LP relaxations of~\eqref{eq:model1}, see~\cite{DPS17}. In general, however, the returned solutions will then not be binary.

\subsection{On the complexity of the partially tomographic case}
\label{sect:partially-tomographic}
We consider the question of when, for $m=2$, the tomographic case can be solved 
efficiently. The following result shows the limitations already for the 
following quite restricted partially tomographic case.
There is only one time step, i.e.,~$t=2$, and~$F^{(1)}$ is  
known while the set $F^{(2)}$ of particle positions for $\tau=2$ is only accessible through
its two X-rays $X_{S_1}F^{(2)}$ and $X_{S_2}F^{(2)}$. 

\begin{thm}\label{thm:tomo-matching}
Even if all instances are restricted to the case $t=2$, where the solution 
$F^{(1)}$ is given explicitly, {\sc TomTrac}${(S_1,S_2)}$ is $\NP$-hard.
Also the corresponding uniqueness problem is $\NP$-complete and the counting problem 
is $\#\P$-complete.
\end{thm}

\begin{proof}
We use a reduction from the following problem, 
where $S_1,S_2,S_3$ are three different lattice lines. 

\begin{problem}{{\sc Consistency}$_{\mathcal{F}^d}{(S_1,S_2,S_3)}$}
\item[Instance] Data functions $f_1,$ $f_2,$ $f_3.$
\item[Question] Does there exist a set $F\in \CF^d$ such that
$X_{S_i}F=f_i$ for all $i\in [3]$?
\end{problem}

This problem and its uniqueness and counting versions have been shown in 
\cite{ggp-99} to be $\NP$-complete or $\#\P$-complete, respectively. In fact, the proof 
of \cite{ggp-99} reveals, that {\sc Consistency}$_{\mathcal{F}^d}{(S_1,S_2,S_3)}$ 
remains $\NP$-complete even if all instances are restricted to those, where, for two 
directions, say $S_2,S_3$, the nonzero X-rays are all $1$ and for each of 
these X-ray lines the grid $G$ contains exactly two candidate points. 
So, let~$(f_1,f_2,f_3)$ be the data functions of such an instance, 
and let $n=\|f_1\|_{(1)}=\|f_2\|_{(1)}=\|f_3\|_{(1)}$. 

Let $F^{(1)}$ consist of $n$ points $p_1,\dots,p_n$ on a line parallel to $S_1$. Hence the support 
of $X_{S_1}F^{(1)}$ is a single line, and the corresponding value is $n,$ while the 
support of $X_{S_2}F^{(1)}$ consists of $n$ lines, and the corresponding values 
are $1$. Note that $F^{(1)}$ is uniquely determined by its X-rays $X_{S_1}F^{(1)}$ 
and $X_{S_2}F^{(1)}$.

Now, we consider the given (restricted) instance $(f_1,f_2,f_3)$ of 
{\sc Consistency}$_{\mathcal{F}^d}{(S_1,S_2,S_3)}$ as `living' at time $\tau=2$, 
regard the first two data functions as the input $(f_1^{(2)}, f_2^{(2)})$ at $\tau=2$ 
of our dynamic 2-direction X-ray problem, and introduce an objective 
function that allows to model the X-ray information in the third direction via 
a minimum weight bipartite matching of cardinality~$n.$ 

So, let $G^{(2)}$ be the grid coming from the data function $f_1^{(2)}=f_1$ 
and $f_{2}^{(2)}=f_2.$ 
Further, let $T_1,\ldots,T_n$ be the lines parallel to $S_3$ that meet the 
candidate grid $G$ of $(f_1,f_2,f_3),$ and let
$g_{i,j}$, $i=1,2$, denote the corresponding two grid points on $T_j$, $j\in [n]$.
Finally, for $j\in[n]$ and $g\in G^{(2)},$ we define the objective function vector~$w= (\omega_{p_j,g})$ by setting
\[
\omega_{p_j,g}= \begin{cases} 
0 & \mbox{for $g \in \{g_{1,j},g_{2,j}\}$;}\\
1 & \mbox{otherwise.}
\end{cases}
\]
Now consider a minimum weight bipartite matching of cardinality~$n$ on the complete bipartite graph with 
the partition~$(F^{(1)},G^{(2)})$ of the vertex set $F^{(1)}\cup G^{(2)}$ and edge costs $\omega_{p_j,g}$ for any edge $\{p_j,g\}\in F^{(1)}\times G^{(2)}.$  In this matching
 let $F^{(2)}$ denote the endpoints of the tracking edges
that are not in $F^{(1)}.$ Then, of course, each point 
$p_j\in F^{(1)}$ is assigned to exactly one point of~$G^{(2)}.$ 
If the objective function value of the given matching is $0$ then
$p_j$ must be assigned to~$g_{1,j}$ or~$g_{2,j}$ for each $j\in [n].$ Hence $X_{S_3}F^{(2)}=f_3.$

On the other hand, if $F^{(2)}$ satisfies $X_{S_i}F^{(2)}=f_i,$ for $i\in [3],$ then, 
in particular, for each $j\in [n]$, the set $F^{(2)}$ contains exactly one of the 
points $g_{1,j}$ or $g_{2,j}$.
Hence there is a matching of cardinality $n$ with objective function value $0$.

Thus, the given instance $(f_1,f_2,f_3)$ is feasible if, and only if, 
there exist a set $F^{(2)}$ satisfying the X-ray constraints with respect to $S_1$ and $S_2$ that allows a perfect matching between $F^{(1)}$ and $F^{(2)}$ of weight~$0.$
Of course, the transformation runs in polynomial-time and is parsimonious.
\end{proof}

The following corollary is merely a reformulation of Theorem ~\ref{thm:tomo-matching} highlighting its interpretation in terms
of matchings.

\begin{cor}\label{cor:tomo-matching}
{\sc Min weight max cardinality bipartite Matching} under totally unimodular 
constraints for one set of the bipartition is $\NP$-hard.
Also the corresponding uniqueness problem is $\NP$-complete while the counting 
problem is $\#\P$-complete.
\end{cor}

It is clear that Theorem \ref{thm:tomo-matching} can also be interpreted as 
a non-approximability result. In fact, since, in the case of feasibility, 
the optimal objective function value is $0$, the $\NP$-hardness of the
problem trivially means that, unless $\P=\NP$, there is no polynomial-time 
approximation algorithm of relative accuracy at most $\alpha$ for any 
factor $\alpha$. Note, however, that the definition of the objective function 
in the proof of Theorem \ref{thm:tomo-matching} can be altered in order to 
guarantee that the objective function is bounded away from~$0.$ 
Simply set $\omega_{p_j,g}=1$ for $j\in[n]$ and $g \in \{g_{1,j},g_{2,j}\},$ and let all other values 
be sufficiently large. This fact can be used to interpret the previous result
in the context of approximation complexity even if we assume that the optimum 
is bounded away from $0$ from below. 

\begin{cor}\label{cor:tomo-approx}
Unless $\P=\NP$, there is no polynomial-time algorithm that solves the 
following problem up to any polynomial-space multiplicative constant in 
polynomial-time.
Given $F^{(1)}$, $f_1^{(2)},$ $f_2^{(2)},$ let $G^{(2)}$ be the grid associated with 
$f_1^{(2)},$ $f_2^{(2)},$ and let $w\ge 1$  encode the edge weights of the complete
bipartite graph with bipartition $(F^{(1)},G^{(2)}).$ Find a cardinality $|F^{(1)}|$ 
matching of $F^{(1)}$ and a set $F^{(2)}\subseteq G^{(2)}$ with
$X_{S_1}F^{(2)}=f_1^{(2)},$ $X_{S_2}F^{(2)}=f_2^{(2)},$ that is minimal 
with respect to $w$ for all such sets $F^{(2)}.$
\end{cor}

In view of the comments after Corollary \ref{cor:tomo-matching}, 
this result may seem somewhat artificial. It does, however, allow for 
non approximability results that are rather close to the observed
practical difficulty. In fact, the objective function in the proof of 
Theorem~\ref{thm:tomo-matching}, and hence for obtaining the result of
Corollary~\ref{cor:tomo-approx}, can be adapted so as to 
accommodate different experimental settings. 
This can be combined with different choices of $F^{(1)}$.
For instance, if we choose the points $p_j\in F^{(1)}$ in such a way that they 
have equal Euclidean distance to $g_{1,j}$ and $g_{2,j}$ 
but different distance to all other grid points, then the objective function can be
interpreted as rewarding conformity with certain known velocities of the particles. 
Hence the hardness proof does to a certain extent reflect practical 
difficulties in tracking the particles.

We remark that, in contrast to Corollary~\ref{cor:tomo-approx}, there are several approximability results for the positionally determined case and specific types of cost functions, see, e.g., \cite{multitracking5, multitracking11, bottleneckmatching, multitracking3}.

\subsection{On the complexity of the tomographic case when the displacement field is known}
\label{sect:tomographic}
Let us now turn to the other extreme, assuming that the $F^{(1)},\dots,F^{(t)}$ are unknown, but the \emph{particle displacement field} is given for all $\tau\in[t].$ 
Note that this is quite different from the setting considered in the proof of Theorem~\ref{thm:tomo-matching}. In fact, in the proof of Theorem~\ref{thm:tomo-matching} we make explicit use of the fact that there are two possible next positions for each particle, which renders the problem $\NP$-hard. If, on the other hand,~$F^{(1)}$ and the displacement field are known, then $F^{(2)},\dots,F^{(t)}$ and the particle tracks can be determined in polynomial time. In general, however, we will show that the reconstruction problem with a given displacement field is again $\NP$-hard.

Particle displacement fields are often (approximatively) known in experimentally controlled environments (for instance, when charged particles in a particle accelerator or plasma fusion device move in the presence of experimentally controlled electromagnetic fields; see, e.g.,~\cite{plasmaphysicsbook, acceleratorbook}). The particle tracking task reduces in this case to that of reconstructing the particle positions from the measurements. Of course, the general aim is to reduce ambiguities by taking several time steps into account. For instance, in~\cite{DPS17} this problem is considered for~$t=2$ in the context of \emph{compressed motion sensing}.  Our next theorem shows, however, that, unless $\mathbb{P}=\mathbb{N}\mathbb{P},$ there is generally no algorithm that can make efficient use of the additional data.

In order to state the theorem we introduce the following notion. Given $S_1,S_2\in\mathcal{S}^d,$ we call a displacement field $(p,\varphi(p)-p),$ $p\in\Q^d,$ \emph{proper (for $S_1$ and $S_2$)} if $\varphi$ is an affine transformation with $\varphi(\Q^d)\subseteq \Q^d$ and $\varphi(S_i)-\varphi(0)\not\in\{S_1,S_2\},$ $i=1,2.$ Note that, in particular, this notion of properness encompasses several different types of circular displacements; for instance, $\varphi$ might represent a $45^\circ$-rotation in the~$(S_1,S_2)$-plane that is followed by a factor $2/\sqrt{2}$ dilation.

\begin{thm}\label{thm:tomo-circular}
Let $\varphi$ denote the affine transformation (encoded by a rational matrix and a rational translation 
vector) of a proper displacement field. Then, the problem  {\sc TomDisplaceTrac}${(\CO; S_1,S_2)}$
 is $\NP$-hard, even if all instances are restricted to the case~$t=2$.  
Also the corresponding uniqueness problem is $\NP$-hard and the counting problem 
is $\#\P$-hard.
\end{thm}

\begin{proof}
We construct a parsimonious transformation from {\sc Consistency}$_{\mathcal{F}^d}{(S_1,S_2,S_3,S_4)}$  for certain lattice lines $S_1,S_2,S_3,S_4$.
This problem is the analogue to {\sc Consistency}$_{\mathcal{F}^d}{(S_1,S_2,S_3)}$  and involves four 
given data functions $f_1,$ $f_2,$ $f_3,$ $f_4.$ The task is to decide whether there exist a 
set~$F\in \CF^d$ such that $X_{S_i}F=f_i$ for all $i\in [4].$

Again, this problem and its uniqueness and counting versions have been shown in 
\cite{ggp-99} to be $\NP$-complete or $\#\P$-complete, respectively
for {\em any} four different different lattice lines $S_1,S_2,S_3,S_4$. 

Now, let $S_1,S_2\in\mathcal{S}^d$ with  $S_1\neq S_2,$, and set
$S_3=\varphi(S_1)-\varphi(0)$ and $S_4=\varphi(S_2)-\varphi(0)$. Since the given displacement field is
proper, $S_1,S_2,S_3,S_4$ are four different lattice lines. Hence, the hardness results of \cite{ggp-99} 
do apply to {\sc Consistency}$_{\mathcal{F}^d}{(S_1,S_2,S_3,S_4)}$. So, suppose, we are
given an instance of this problem by means of the X-ray functions $(f_1,f_2,f_3,f_4)$, 
and let $G$ be the corresponding candidate grid.

By associating a binary variable with each point $g_1,\dots,g_{|G|}$ of~$G$ we can formulate the task as deciding feasibility of the ILP
\[\begin{array}{lll}
Ax&=&b\\
A'x&=&b'\\
x&\in&\{0,1\}^{|G|},
\end{array}
\] where $A$ and $A'$ encode the point-line incidences along $\{S_1,S_2\}$ and $\{S_3,S_4\},$ respectively, while~$b$ and~$b'$ encode the data functions~$f_1,f_2,$  and $f_3,f_4,$ respectively. 

Setting $A^{(1)}=A,$ $A^{(2)}=A',$ $b^{(1)}=b,$ $b^{(2)}=b',$ this problem is, of course, equivalent to the ILP
\begin{equation}\label{eq:track}
\begin{array}{rcl}
A^{(1)}x^{(1)}  & = & b^{(1)},\\
A^{(2)}x^{(2)}      & = & b^{(2)},\\
x^{(1)}- x^{(2)}  & = & 0,\\
         x^{(1)}  & \in &\{0,1\}^{|G|},\\
         x^{(2)}    & \in &\{0,1\}^{|G|}.
\end{array}
\end{equation}  Now, suppose we associate with each variable $\xi_i^{(2)},$ $i\in[|G|],$ the point $\varphi^{-1}(g_i)\in\Q^d.$ This does not change the ILP, but only its geometric interpretation. The problem $A^{(2)}x^{(2)}=b^{(2)}$ can now be viewed as to involve only X-rays taken from the directions 
\[
\varphi^{-1}(S_{i+2})-\varphi^{-1}(0)= \varphi^{-1}(\varphi(S_i)-\varphi(0))-\varphi^{-1}(0)=S_i, \qquad 
(i\in [2]).
\] 

Consider now the ILP~\eqref{eq:model1} for $t=2$ and the following settings. Let $G^{(1)}$ denote the grid defined by the data functions $f_1,f_2$ (along~$S_1,S_2$), and let $G^{(2)}=\varphi^{-1}(G')$ where $G'$ denotes the grid defined by the data functions $f_3,f_4$ (along~$S_3,S_4$).  Let the points of each grid $G^{(\tau)}$ be labeled $g_i^{(\tau)}$ for $i\in I^{(\tau)}=[|G^{(\tau)}|],$ $\tau\in[t],$ and set
\[
\omega_{i,j}^{(1)}=\left\{\begin{array}{lll}0 && \textnormal{for } \varphi(g_j^{(2)})=g_i^{(1)}\in G, \\ 1&& \textnormal{otherwise,}\end{array} \right.
\] for every $(i,j)\in I^{(1)} \times I^{(2)}.$ 
 
An optimal solution to this ILP~\eqref{eq:model1} has the objective function value~0 if, and only if, the ILP~\eqref{eq:track} is feasible. As the transformation runs in polynomial-time and is parsimonious, we have thus completed the proof of this theorem.
\end{proof}

While Theorem~\ref{thm:tomo-circular} shows the general intractability, there are specific types of displacements for which the corresponding problem {\sc TomTrac}${(S_1,S_2)}$ turns out to be in $\mathbb{P}.$ This is trivially the case if the particle displacement is known to be zero. Examples, some non trivial, will be given in Theorem~\ref{thm:orth} and~\ref{thm:comb5}. 

\subsection{A rolling horizon model}\label{subsec-rolling-horizion}
Even though Theorem \ref{thm:tomo-matching} was not available then, in 
retrospective the rolling horizon approach of \cite{agms-15} for  
dynamic discrete tomography can be viewed as an attempt to `bypass' the 
$\NP$-hardness result of Theorem~\ref{thm:tomo-matching}. 

In fact, the core of \cite{agms-15}'s polynomial-time algorithm is based 
on a different modeling for the case $t=2$. Unlike \eqref{eq:model1}, the quality of the
tracking is not encoded in terms of weights on the tracking edges, 
but rather through weights on the grid points. 

Before we analyze the price that has to be paid for such a simplification let us
recall the relevant details of the rolling horizon algorithm of \cite{agms-15}
as a service to the reader.

The time step from $\tau$ to $\tau+1$ is modeled as a linear program.
It is based on the knowledge of a solution~$F^{(\tau)}$, hence 
in effect dealing with the partially tomographic case. 
The constraints encode the X-rays provided by the data functions
$f_1^{(\tau+1)}$, $f_2^{(\tau+1)}$. The variables correspond to the points in the 
grid $G^{(\tau+1)}$ and are again collected in a vector $x^{(\tau+1)}$.
The X-ray information is encoded by means of a totally unimodular matrix
$A^{(\tau+1)}$ and a right-hand side vector $b^{(\tau+1)}$.

Further, each point~$g_i^{(\tau+1)} \in G^{(\tau+1)}$ carries a weight 
$\alpha_{i}^{(\tau+1)}$ which reflects the `distance' to a closest 
point $g_i^{(\tau)} \in F^{(\tau)}$ (which is a likely `predecessor'). 
Let  $a^{(\tau+1)}$ be the corresponding weight vector. 
Various weight functions are discussed in \cite{agms-15}. For instance,
if the particles move slowly, a reasonable choice for $a^{(\tau)}$ is
(a finite precision approximation of)
\[
\alpha_{i}^{(\tau)}=\min_{j:\xi^{(\tau)}_{j}=1}\bigl\{\|g_{i}^{(\tau+1)}-g_{j}^{(\tau)}\|\bigr\},
\]
where again $\|\cdot\|$ denotes some suitable norm. 
Clearly, this choice of weights prefers positions of particles
close to those from the previous time step.

The algorithm {\sc Rolling Horizon Tomography} is now as follows:

\begin{alg}[{\sc Rolling Horizon Tomography}]\label{alg:rolling-horizon}
Given $F^{(1)}$ perform successively for $\tau =1,\dots,t-1$
\begin{equation}
\begin{array}{lrcl}
&\multicolumn{3}{c}{\min \,\,\bigl(a^{(\tau+1)}\bigr)^T x^{(\tau+1)}}\\[.1cm]
\textnormal{s.\,t. }&A^{(\tau+1)}x^{(\tau+1)} & = & b^{(\tau+1)},\\
&       x^{(\tau+1)} & \le &\1,\\
&       x^{(\tau+1)} & \ge & 0,\\
\end{array}
\label{eq:LP}
\end{equation} 
to determine $F^{(\tau+1)}$ (corresponding to the incidence vector $0$-$1$ of a 
basic feasible solution of the linear program).
Finally the paths $\CP_1,\ldots,\CP_n$ are obtained by a routine
 that computes a perfect bipartite matching in the graph with vertices $F^{(\tau)}\cup F^{(\tau + 1)}$ and 
edges corresponding to the pairs of vertices which realize the distances~$\alpha_{i}^{(\tau)}$.
\end{alg}

The rationale behind the rolling horizon model seems quite similar to 
that analyzed in Theorem \ref{thm:tomo-matching}. The two models are,
however, fundamentally different. This is obvious from the fact
that the former can be handled in polynomial time while the latter leads
to an $\NP$-hard problem. So, why are we not fully satisfied with 
{\sc Rolling Horizon Tomography}?
After all, it is exact in the sense that it is guaranteed to yield a solution 
which matches the data. It also allows to incorporate physical knowledge and 
is reported to work quite well in practice (see \cite{agms-15,glidingarc-15}). 
There is, however, a price to pay. In fact, as the following example shows,
this algorithm is only a heuristic. In general, one cannot expect to obtain 
the most realistic paths that way. The reason is that the weights used
to measure the quality of an assignment do not incorporate the requirement 
that no two particles can have originated from the same location at the previous moment in time.
Of course, Theorem \ref{thm:tomo-matching} shows that, unless $\P=\NP$, 
there will never be a full remedy, i.e., there does not exist a polynomial-time 
algorithm that is also exact in that second sense even in this restricted case.

\begin{ex}\label{ex:grid1}
Figure \ref{fig:grid1} shows an example of $n=3$ points at time $\tau=1$ in the plane
(black dots) and the grid obtained from the X-ray 
measurements in the coordinate directions at time $\tau=2$. The weight
of each of the nine grid points is given by its Euclidean distance to the 
closest black point.

\begin{figure}[htb]
\begin{center}
\subfigure[]{\includegraphics[width = .3\textwidth]{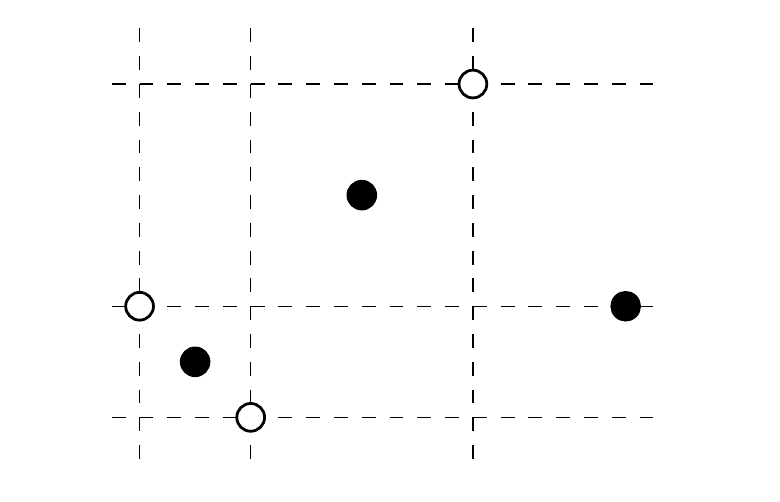}}\hspace*{15ex}
\subfigure[]{\includegraphics[width = .3\textwidth]{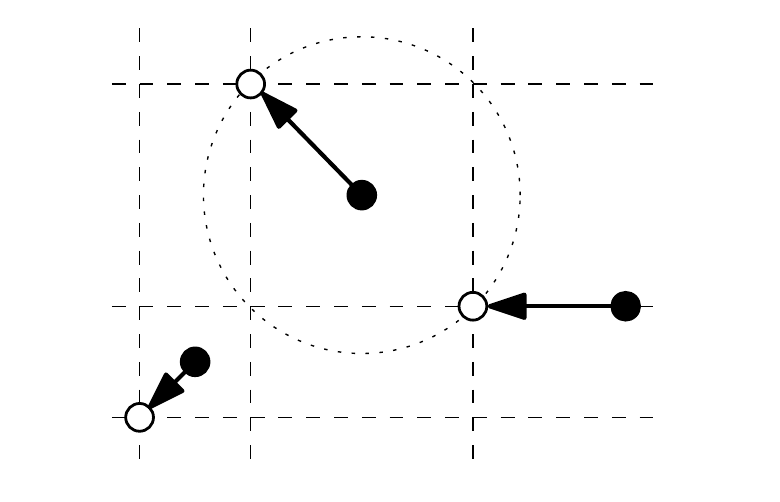}}
\caption{The black dots indicate the solution at time $\tau=1$.
The white dots in~(a)~give a solution at time $\tau=2$ found by
{\sc Rolling Horizon Tomography}, i.e., with the LP \eqref{eq:LP},
while the white dots in~(b) show the solution at time $\tau=2$ that is closest to that
at time $\tau=1$ (the arrows indicate the displacement, 
each within a radius equal to that of the dotted circle). }
\label{fig:grid1}
\end{center}
\end{figure} 

The white dots on the left give a solution $F^{(2)}$ obtained by the rolling horizon~{\sc Lp}~\eqref{eq:LP} for these weights. The best (and quite different) assignment is
depicted on the right. The reason for the failure of the algorithm is that
the weights assigned to the four grid points~$g_1,g_2,g_3,g_4$ 
in the left lower corner all stem from just one black point. 
Since the top black point has minimal distance to four grid points (three of which 
are different from the $g_i$), any minimal solution of \eqref{eq:LP} does 
contain two of the grid points $g_1,g_2,g_3,g_4$. Obviously,
the particle paths indicated in the right-hand side figure are shorter.
\end{ex}

\subsection{Models involving particle history}\label{sect:nonmarkov}
The models in the previous subsections are  
{\em memory-less} in the sense that the estimation of where a particle $p\in P$ 
should move in the time step from $\tau$ to $\tau+1$ depend only 
on its position at time $\tau$ but not on other, earlier positions.

In view of the tractability and intractability results of 
Theorem \ref{thm:ILP-posdet} and Theorem \ref{thm:tomo-matching}, 
respectively, let us first consider the coupling in the positionally determined case
for $t\ge 3.$  

Of course, if we are interested in particle tracks $\CP_1,\ldots, \CP_n$
whose total length 
$$
c(\CP_1,\ldots, \CP_n)= \sum_{i\in [n]} \sum_{\tau\in [t]} \sum_{e\in \CP_i} w(e) 
$$
is minimal, we can find one in polynomial time by Theorem \ref{thm:ILP-posdet}.
The following example shows, however, that such an optimal 
assignment of positions to particles for every moment in time does not necessarily 
guarantee `overall reasonable' particle tracks even if there are only two time steps involved.

\begin{ex}\label{ex:nohistory}
Figure~\ref{fig:nohistory} shows for the positionally determined case two particles (black dots) at times~$\tau=1,2,3$ (dashed lines). Figure~\ref{fig:nohistory} (a)
is obtained by assigning each of the two points at time $\tau$ to their
nearest neighbor at time $\tau+1$, for $\tau=1,2$.
This is optimal if we are aiming at paths of shortest total length but leads 
to a kink at~$\tau=2.$ 
Figure \ref{fig:nohistory} (b) depicts slightly longer paths, which, on the other hand,
correspond to movements along straight lines.

\begin{figure}[htb]
\begin{center}
\subfigure[]{\includegraphics[width = .3\textwidth]{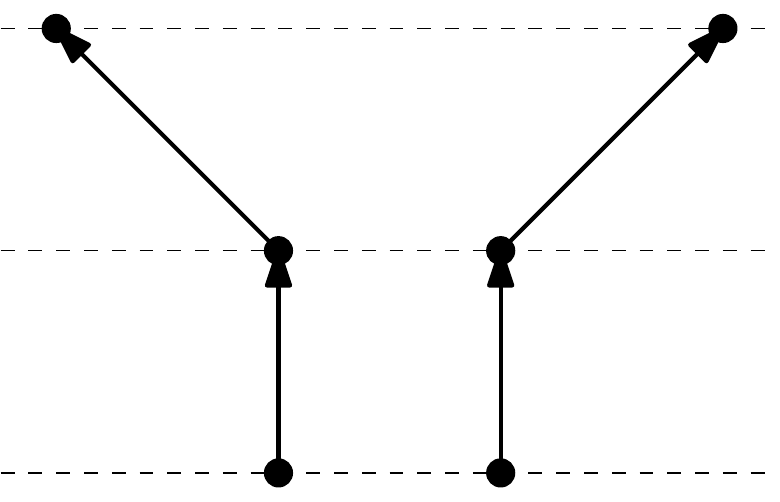}}\hspace*{15ex}
\subfigure[]{\includegraphics[width = .3\textwidth]{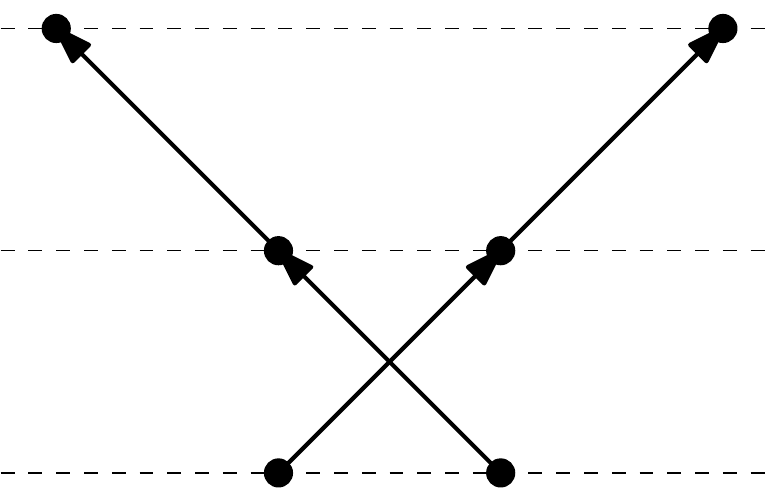}}
\caption{The black dots indicate the solutions $F^{(\tau)}$ at times $\tau=1,2,3$. Depicted are: 
(a) particle paths with the shortest total length, and (b) straight line particle paths.}
\label{fig:nohistory}
\end{center}
\end{figure} 

It is not clear upfront which of the two particle tracks might be more realistic. 
This assessment will certainly depend on prior knowledge.
The case at the left might be more realistic if there is, say, a magnetic force affecting
the particles at $\tau=2,$ which pushes them apart. If there is no indication of
an external force, then the paths on the right appear to be more appropriate.
\end{ex}

It seems quite natural to aim at incorporating general prior knowledge about
`reasonable' paths in the following way. First an expert would compile
a list of most appropriate paths or give at least parametric descriptions of
reasonable trajectories; then an algorithm should compute solutions whose paths do not
deviate too much from a closest one from the list. 

It turns out that there are serious limitations to this approach.

\begin{thm}\label{thm:matching}
The problem {\sc Trac}$(\CO;d)$ is $\NP$-hard, even if all instances are restricted to a fixed $t>2,$ 
and $\CO$ is a path value oracle.

The $\NP$-hardness persists if the objective function values 
provided by $\CO$ are all encoded explicitly.
\end{thm}

\begin{proof} It suffices to prove the result for $t=3$.
We use a transformation from the $\NP$-complete problem

\begin{problem}{{\sc $3$D-Matching}}
\item[Instance] $n\in \N$, finite disjoint sets $X$, $Y$, $Z$, and $W\subseteq X \times Y \times Z$. 
\item[Question] Does there exist a set $M\subseteq W$ with $|M|\ge n$ such that the following holds:
Let $p_i=(x_i, y_i, z_i)\in M$, $i=1,2$, $p_1\ne p_2$, then $x_1 \ne x_2$, $y_1\ne y_2$, and $z_1\ne z_2$? 
\end{problem}

In fact, {\sc $3$D-Matching} was already contained in Karp's original list \cite{karp72} 
of $\NP$-complete problems and was shown there to be $\NP$-complete even for restricted instances
where $n=|X|=|Y|=|Z|$. Further note that the $\NP$-completeness persists if we assume in addition that $X,Y,Z \subseteq \Z^d$ for some fixed~$d\in \N.$ Further details on the computational complexity and connections to assignment problems can be found in \cite{gareyjohnson} and \cite{BDM09}, respectively. For additional results on variants of {\sc $3$D-Matching} with different cost functions see~\cite{multitracking5, multitracking11, multitracking9, multitracking3, multitracking8}; for a polynomial-time solvable variant, where the cost function has the so-called \emph{Monge property}, see~\cite{mongeproperty} (for applications, see~\cite{mongeapplications}). 

Now, suppose we are given such a (restricted) instance $(n;X,Y,Z;W)$ of {\sc $3$D-Matching}.
We set $F^{(1)}=X,$ $F^{(2)}=Y,$ $F^{(3)}=Z.$ The matching condition means
that we need to select $n$ disjoint paths~$\CP_1,\dots,\CP_n.$

The set $W$ will now be encoded in a corresponding instance of {\sc Trac}$(\CO;d)$ by means of an objective function~$c$ defined on the
set of all paths. For every path~$\CP=(g^{(1)}, g^{(2)}, g^{(3)}),$ where 
$g^{(\tau)} \in F^{(\tau)}$ for $\tau\in [3],$ we define
\[
w(\CP)= \begin{cases} 
0 & \mbox{if $(g^{(1)}, g^{(2)}, g^{(3)})\in W$,}\\
1 & \mbox{otherwise.}\\
\end{cases}
\] The cost $c(\CP_1,\dots,\CP_n)$ of any $n$ paths $\CP_1,\dots,\CP_n$ is given by $c(\CP_1,\dots,\CP_n)=\sum_{i=1}^nw(\CP_i).$

Then, there exists a matching of cardinality $n$ in $W$ if, and only if, 
there exists a solution to the {\sc Trac}$(\CO;d)$ instance with cost~$0.$
Of course, the transformation runs in polynomial time.

The second statement follows from the fact that there are only $O(n^t)$ different
paths whose costs~$0$ or~$1$ need to be encoded. 
\end{proof}

Let us point out that Theorem \ref{thm:matching} applies to the 
situation that an expert has provided an explicit list of all paths 
that are regarded physically reasonable, and the task is simply to determine whether
there exists a solution $(\CP_1,\ldots,\CP_n)$ that consists
entirely of paths from that list. 
Again we can adapt the objective function in the proof of Theorem \ref{thm:matching} 
to show that it is also an $\NP$-hard task to find particle tracks that are within a certain specified 
distance from tracks of such a list.

The following theorem makes use of a result of \cite{multitracking10} on a variant of {\sc $3$D-Matching}.

\begin{thm}\label{thm:special-weights}
For $d=2$ and $t=3$ the $\NP$-hardness of Theorem \ref{thm:matching} persists 
if the weight of $\CP=(g^{(1)}, g^{(2)},g^{(3)})$ is given by
\begin{equation} \label{eq:triangle}
w(\CP)=\frac{1}{2} |\det (g^{(2)}- g^{(1)}, g^{(3)}- g^{(1)})|.
\end{equation} 
Even checking whether a solution of weight $0$ exists is $\NP$-complete. 
\end{thm}

\begin{proof}
The assertion follows from the result of \cite{multitracking10} that states that the problem
 {\sc A3ap} is $\NP$-hard (and its respective decision version is $\NP$-complete). Here disjoint sets $X,Y,Z \in \Z^2$ of cardinality $n$ are given,
and the goal is to find $n$ disjoint subsets $\{x,y,z\}$ with $x\in X$, $y\in Y$, $z\in Z$ of minimal weight.
For {\sc A3ap} the weight is defined as the sum of the areas of the triangles.
Now, recall that the area of the triangle $\conv\{x,y,z\}$ is given by 
$\frac{1}{2}|\det (y-x, z-x)|$.
\end{proof} 

\begin{cor} \label{cor:straightline}
For every fixed $d\geq 2$ and $t\geq 3$ it is an $\mathbb{N}\mathbb{P}$-complete problem to decide whether a solution of {\sc Trac}$(\CO;d)$ exists where all particles move along straight lines.
\end{cor}
\begin{proof} The result for $d=2$ and $t=3$ was given in Theorem~\ref{thm:special-weights}. As every planar instance can be viewed as an instance in $\Q^d,$ $d\geq2,$ the $\NP$-hardness carries immediately over to any $d\geq 2.$ Further, as any instance for $t=3$ can be viewed as a special case of an instance for $t\geq3$ where the particles do not move after $\tau=3,$ the $\NP$-hardness persists for any $t\geq3.$ 
\end{proof}

Clearly, for non-negative weights $w(\CP_k)$ we have \[\sum_{k=1}^n w(\CP_k)=0 \qquad \Leftrightarrow \qquad \max_{k\in[n]} w(\CP_k)=0.\] Hence, Theorem~\ref{thm:special-weights} implies the following result (see also~\cite{bottleneckmatching}).

\begin{cor} \label{cor:minmax}
Already for $d=2$ and $t=3$ it is $\NP$-hard to find a solution of {\sc Trac}$(\CO;d),$ which minimizes the objective function 
\[c(\CP_1,\dots,\CP_n)=\max_{i\in[n]}w(\CP_k).
\] 
\end{cor}

Comparing the results from Theorem~\ref{thm:special-weights} 
with that from Theorem~\ref{thm:ILP-posdet} it seems interesting to note that the former problem is~$\NP$-hard while its Markov-type counterpart is polynomial-time solvable. Results similar to Theorem~\ref{thm:special-weights} for other types of weights can be found in~\cite{multitracking9,multitracking10}.

In view of the discouraging results above one cannot expect to be able to incorporate too much 
of the individual particles'  history or information about reasonable trajectories without sacrificing efficiency.
In fact, unless $\P=\NP$, any algorithm for that task must either fail 
to produce an optimal solution (at least under specific circumstances) or must have a super-polynomial running time. In other words, any such algorithm 
bears, for a given practical instance, the risk of not returning a solution within an
acceptable time. 

With this warning we introduce new polynomial-time heuristics, which favor 
paths that are regarded `reasonable.' The algorithms can be seen as examples 
of more general paradigms. In order to keep the exposition simple we describe 
only basic prototypes.

Let us point out first that, of course, the rolling horizon Algorithm \ref{alg:rolling-horizon} can be 
modified in such a way that it allows to utilizes even `intuitive' knowledge 
of how the current particles tracks determined up to the moment~$\tau$ in time 
should possibly be extended to $\tau+1$.
All that is needed is to quantify this knowledge and use it to define 
the weights for the next time step. If the previous knowledge is restricted
to a fixed number~$k$ of the last past moments we arrive at a {\sc $k$-Rolling Horizon}
algorithm. This may be a promising method if the correct sets 
$F^{(1)},\ldots,F^{(k)}$ are known. 

If this is not the case, {\sc $k$-Rolling Horizon} may run into problems.
In fact, for the tomographic construction of the first set $F^{(1)}$ 
no other than the X-ray information is available. 
The determination of $F^{(2)}$ is then also based on $F^{(1)}$ etc. 
But this means that the choices in the previous steps cannot be revised later
anymore. Hence there may be other sets, tomographically 
equivalent to the choice of $F^{(1)}$ made by the algorithm, which allow 
much more realistic paths.
If, at some moment $\tau$ the deviation from
realistic shapes is noticed, {\sc $k$-Rolling Horizon} does, however, not
provide any remedy.

We will therefore pursue now an approach that is capable of using
background knowledge in a more global and balanced way.
We begin by introducing a mathematical setting for regarding a 
particle path~$\CP=(g^{(1)},\ldots, g^{(t)})$ 
of $\CI=(G^{(1)},\ldots, G^{(t)})$ as {\em reasonable}.

Let, again, $\|\cdot\|$ be a norm in $\R^d,$ let $h:[0,\infty[\rightarrow [0,\infty[$ be 
strictly monotone such that $h(\|x\|)\in \Q$ for every $x\in \Q^d$ and such that $\size(h(\|x\|))$ 
is bounded by a polynomial in $\size(x).$ Examples for pairs~$(\|\,\cdot\,\|,h)$ include~$(\|\,\cdot\,\|_{(\infty)},\id)$ and $(\|\,\cdot\,\|_{(p)},(\,\cdot\,)^p)$ for $p\in \N.$

For $k\in \N$, any $k$-tuple $C=(g^{(\tau_1)},\ldots, g^{(\tau_k)})$ with $1\leq \tau_1<\tau_2< \ldots <\tau_k\leq t$
and $g^{(\tau_i)}\in G^{(\tau_i)},$ for~$i\in [k],$ 
is called a {\em $k$-sample} of $\CI$. Let $\CC(\CI)$ denote the set of
all $k$-samples of $\CI$. Further, if $Q$ is a set of grid points, we
write $\CC(\CI;Q)$ for the subset of those $C$ that contain all elements of~$Q.$
Also, if the $\tau_i\in [t]$ are fixed we speak of a $k$-sample for 
$(\tau_1,\tau_2,\ldots,\tau_k)$. Let $\CC(\CI; Q; \tau_1,\ldots,\tau_k)$ 
denote the set of all $k$-samples of $C\in \CC(\CI;Q)$ for $(\tau_1,\ldots,\tau_k)$.

Now, let $k\in \N$ and let $\CR_k(\CI)$ be a family of curves $r:[0,t]\rightarrow \Q^d$
with the following properties:
For any $k$-sample $C=(g^{(\tau_1)},\ldots, g^{(\tau_k)})$ of $\CI$ 
there is a unique $r_C\in \CR_k(\CI)$ such that $r_C(\tau_i)= g^{(\tau_i)}$ for $i\in [k]$.
Further, $r_C(\tau)$ and $h(\|r_C(\tau)- g^{(\tau)}\|)$ can be computed in 
polynomial time for $\tau \in [t]\setminus \{\tau_1,\ldots,\tau_k\}$.
The function $r$ will be called {\em sample fit}. 

As an example, we may choose $\CR_2(\CI)$ to consist of all lines through 
two grid points at two different moments in time. 
Such a choice would favor straight line movements of
particles. Generalizing this, we  may consider $\CR_k(\CI)$ to consist of all polynomial curves~$\varphi$ of degree~$k-1$ or less, i.e.,  \[\varphi :\Q\to\Q^d, \qquad \tau\mapsto a_0+a_1\tau+\cdots+a_{k-1}\tau^{k-1},\] with $a_1,\dots,a_{k-1}\in\mathbb{Q}^d.$ Quadratic curves, for instance, are often used to describe trajectories of objects that move under the action of gravity; see, e.g.,~\cite{basketballtracking, dropletmotion}. 

The sample fits will be regarded as representing the a priori knowledge 
about the shapes and other properties of the particle paths. This
knowledge can be incorporated in various ways. Let us begin with
the positionally determined case, i.e., we assume $G^{(\tau)}=F^{(\tau)}$ for $\tau\in [t].$

So, suppose, that for some fixed $k\in \N$ sample fits 
$\CR_k(\CI)$ are (implicitly) available as 
specified above. Then the following algorithm prefers particle tracks 
that are close to curves of $\CR_k(\CI).$

\begin{alg}[{\sc Path Fitting}]\label{alg:curve} Let
$\CI=(F^{(1)},\ldots, F^{(t)})$ be given. Choose $\tau_1,\dots,\tau_k$ with $1\leq\tau_1<\tau_2< \ldots <\tau_k\leq t.$
Then, for every $i \in I^{(\tau_1)}$ and $j \in I^{(\tau_k)},$ set   $Q_{i,j}=\{g_i^{(\tau_1)},g_j^{(\tau_k)}\}$ and compute 
\begin{equation} \label{eq:sample1}
\gamma_{i,j}=
\min_{C\in \CC(\CI; Q_{i,j}; \tau_1,\ldots,\tau_k)} \, 
\max_{\tau \in [t]} \, 
\min_{g^{(\tau)} \in G^{(\tau)}} h(\|r_C(\tau)- g^{(\tau)}\|).
\end{equation}
Next compute a minimum weight perfect bipartite matching $M$ for $G^{(\tau_1)}$, $G^{(\tau_k)}$ 
with weights~$\gamma_{i,j}$ for the edges~$(g_i^{(\tau_1)}, g_j^{(\tau_k)})$, $i \in I^{(\tau_1)},$ $j \in I^{(\tau_k)},$
and with~$r_{i,j}$ denoting a corresponding sample fit for which the minimum in \eqref{eq:sample1} is attained.
Finally, assign to the $n$ curves $r_{i,j},$ $(g_i^{(\tau_1)}, g_j^{(\tau_k)})\in M,$ the points of $F^{(\tau)}$ for $\tau \in [t]\setminus \{\tau_1,\ldots,\tau_k\}$ according to the (lexicographical) nearest neighbor rule with respect to their
reference points $r_{i,j}(\tau).$
\end{alg}

As described, there is still ample freedom for specifying or varying Algorithm \ref{alg:curve}. We can, for instance, replace the weights ~\eqref{eq:sample1}
in various ways. For instance, another natural choice is  
\[
\gamma_{i,j}=
\min_{C\in \CC(\CI; Q_{i,j}; \tau_1,\ldots,\tau_k)} \, 
\sum_{\tau \in [t]\setminus \{\tau_1,\ldots,\tau_k\}} \, 
\Bigl(\min_{g^{(\tau)} \in G^{(\tau)}} h(\|r_C(\tau)- g^{(\tau)}\|\Bigr)^2.
\]
Another option is to average over a prescribed number of near best 
sample fits. Further, the exact condition 
$r_C(\tau_i)= g^{(\tau_i)}$ for $i\in [k]$ can be replaced by an approximate sample fit. 
Hence {\sc Path Fitting} can be considered as a class of algorithms,
and the various specifications need to be comprehensively evaluated 
on real-world data. In any case, the framework is algorithmically efficient.

\newpage

\begin{thm}\label{thm:curve-fitting}
{\sc Path Fitting} runs in polynomial time.
\end{thm}

\begin{proof}
Simply observe that we need $n^2\cdot n^{k-2}\cdot (t-k) \cdot n = O(n^{k+1}t),$ i.e., polynomially many, computations
 to determine all $\gamma_{i,j},$ $i,j\in [n].$ Further, all other computations 
can be performed in polynomial time.
\end{proof}

Note that for $k=2$ and lines as sample fits, when applied to
Example \ref{ex:nohistory}, {\sc Path Fitting} produces
the solution shown in Figure~\ref{fig:nohistory}(b).
In general, however, it follows from Theorem \ref{thm:special-weights} that, unless $\P=\NP,$ 
{\sc Path Fitting} will not always return optimal solutions
even with lines as sample fits. In fact, the existence of lines as
sample fits corresponds to the weights \eqref{eq:triangle} in Theorem \ref{thm:special-weights}.

It is possible to use the above framework also in the fully tomographic case,
particularly, to reduce the ambiguity in the determination of the tomographic 
solutions $F^{(\tau)}$ for $\tau\in [t]$. 

\begin{alg}[{\sc Tomographic Fitting}]\label{alg:tomfit} Let, for $\tau\in [t],$ 
data functions $(f_1^{(\tau)}, f_2^{(\tau)})$ be given with~$G^{(\tau)}$ denoting the corresponding grid. Further, let 
$\CI=(G^{(1)},\ldots, G^{(t)}).$ 

For $\tau\in [t],$  $g_i^{(\tau)}\in G^{(\tau)}$ and
$Q_i^{(\tau)}=\{g_i^{(\tau)}\}$ set
\begin{equation} \label{eq:sample}
\alpha(g_i^{(\tau)})=
\min_{C\in \CC(\CI; Q_i^{(\tau)})} \, 
\max_{\tau' \in [t]} \, 
\min_{g^{(\tau')} \in G^{(\tau')}} h(\|r_C(\tau')- g^{(\tau')}\|).
\end{equation}
Now, with $a^{(\tau)}$ being the vector with components $\alpha(g_i^{(\tau)})$
and  $A^{(\tau)}$, $x^{(\tau)}$, and $b^{(\tau)}$ as before,
solve the~$t$ uncoupled linear programs
\begin{equation}
\begin{array}{lrcl}
&\multicolumn{3}{c}{\min \,\,\bigl(a^{(\tau)}\bigr)^T x^{(\tau)}}\\[.1cm]
\textnormal{s.\,t. }&A^{(\tau)}x^{(\tau)} & = & b^{(\tau)}\\
&       x^{(\tau)} & \le &\1\\
&       x^{(\tau)} & \ge & 0\\
\end{array}
\label{eq:fitting-LP}
\end{equation} 
to determine sets $F^{(\tau)}$ which are consistent with the 
X-ray information.
Finally the paths $\CP_1,\ldots,\CP_n$ are obtained by $t-1$
calls to a routine for minimum weight perfect bipartite matching 
with weights $\omega_{i,j}^{(\tau)}= \alpha(g_i^{(\tau)})+\alpha(g_j^{(\tau+1)})$.
\end{alg}

There are, again, various other reasonable ways to define the weights $\alpha$ and $w.$ 

Of course, the final matching part in {\sc Tomographic Fitting} can be replaced by  
{\sc Path Fitting}. This leads to the algorithm
{\sc Tomographic Path Fitting}, which uses the former to produce 
$F^{(1)},\ldots,F^{(t)}$ and subsequently the latter to identify the paths.
Clearly, they both run in polynomial time.

\begin{thm}\label{thm:curve-fitting2}
{\sc Tomographic Fitting} and {\sc Tomographic Path Fitting} run in polynomial time.
\end{thm}

Let us now turn to another algorithmic paradigm. 
It can be viewed as a rolling horizon method augmented by 
some `clamping' strategy for improvement. 
For the sake of simplicity, the method will in the following be specified for the  
case $t=3$ and $k=2$ where straight line tracks are desired.

\begin{alg}[{\sc Two-way Fitting}]\label{alg:two-way}
Construct a set $F_1^{(1)}$ which satisfies the tomography constraints, and 
apply Algorithm \ref{alg:rolling-horizon} with a chosen 
objective function to compute $F_1^{(2)}$.
Let $M_1^{1,2}$ denote the pairs~$(i,j)\in I^{(1)}\times I^{(2)}$
such that $g_i^{(1)}$ is matched to $g_j^{(2)}$.

Then, by using velocity information, if available, compute a subset $B_{i,j}$ of the affine hull
$L_{i,j}=\aff \{g_i^{(1)}, g_j^{(2)}\}$ for $(i,j)\in M_1^{1,2}$
(on which, if the line would represent the real particle path, the next point 
would be expected). If no velocity information is available, take
$B_{i,j}=L_{i,j}$. Note that $L_{i,j}$ is a line unless~$g_i^{(1)}= g_j^{(2)}.$

Then, compute for each point $g^{(3)}\in G^{(3)}$ its weight
\[w(g^{(3)})=\min_{(i,j)\in M_1^{1,2}} \min_{x\in B_{i,j}} \|g^{(3)}-x\|_{(2)}^2,\]
and apply again Algorithm \ref{alg:rolling-horizon}, yielding $F_1^{(3)}$ and
a set $M_1^{2,3}\subseteq I^{(2)}\times I^{(3)}$ of matching index pairs.

Next, apply  Algorithm~\ref{alg:tomfit} for $k=2$  and 
$\CR_2(\CI)$ parametrizing the lines
through $F_1^{(3)}$ and $F_1^{(1)}$ to obtain for $\tau=2$ and weights~\eqref{eq:sample} a tomographically feasible
set $F_2^{(2)}$ and, subsequently, apply Algorithm \ref{alg:rolling-horizon} with
weights computed as in the first part of the algorithm to 
obtain a solution $F_2^{(1)}.$ Then, the procedure is iterated.

The algorithm terminates if in the next round the same triple $(F_l^{(1)}, F_l^{(2)}, F_l^{(3)})$
of sets is repeated.
\end{alg}



As pointed out before, none of these algorithms can circumvent the 
$\NP$-hardness of the problem it addresses. Therefore all such methods are either 
only heuristics and may fail to produce a reasonable solution or they have a
super-polynomial running time. This is a worst-case analysis. Nevertheless, the introduced paradigms offer a variety of different approaches that can be tested and compared on any given real-world data set at hand. The application in~\cite{glidingarc-15} demonstrates a case where at least one of the approaches performed successfully on a given real-word data set.

\section{Combinatorial models}\label{sect:combinatorial}

The main part of the algorithm {\sc Tomographic Fitting} was
to reduce the ambiguity that is present in our tomographic tasks 
(particularly for just two directions) by 
utilizing a priori knowledge within an optimization routine. 
In the following we will pursue a similar idea, which, however, is based
now on combinatorial requirements. 

We restrict the exposition to the case $d=m=2$ with $S_1,S_2$ denoting
the coordinate directions in $\R^2.$ Using the same notation as before, the set 
of solutions of the X-ray problems 
\[
\begin{array}{rcllr}
A^{(\tau)}x^{(\tau)} &=&b^{(\tau)} &&(\tau\in[t]),\\
x^{(\tau)} &\in & \{0,1\}^{|G^{(\tau)}|} &&(\tau\in[t]),
\end{array}
\]
will now be  restricted by incorporating prior knowledge about the potential movement of particles as hard combinatorial requirements.
More precisely, we assume that for $\tau\in[t]$ we are given 
$l_\tau\in\mathbb{N}_0$ non-empty {\em windows} 
$$
W^{(\tau)}_1,\dots,W^{(\tau)}_{l_\tau}\subseteq G^{(\tau)}
$$
together with some information about the number of grid points 
contained in them. For each window, this knowledge is presented by a pair 
$$
(\sim^{(\tau)}_i,k_i^{(\tau)})\in\{\leq,=,\geq\}\times \mathbb{N}_0.
$$
If there is no risk of confusion we identify each window $W^{(\tau)}_i$ 
with the set of indices $\{j: g_j^{(\tau)}\in W^{(\tau)}_i\}.$ 
Then the \emph{window constraints} are of the form 
$$
\sum_{j\in W^{(\tau)}_i}\xi_j^{(\tau)}\sim^{(\tau)}_i k^{(\tau)}_i.
$$
Note that the variables $\xi_j^{(\tau)}$ are the same as before 
and hence binary. Therefore the tasks of
{\sc Tomography under Window Constraints} can be formulated as the following \textsc{Ilp}s.
\begin{equation}\label{eq:windowilp}
\begin{array}{rclr}
A^{(\tau)}x^{(\tau)}&=&b^{(\tau)} &(\tau\in[t]),\\
\sum_{j\in W^{(\tau)}_i}\xi_j^{(\tau)}&\sim^{(\tau)}_i&k^{(\tau)}_i, &(\tau\in[t],\: i\in[l_\tau]),\\
x^{(\tau)}&\in&\{0,1\}^{|G^{(\tau)}|} &(\tau\in[t]).
\end{array}
\end{equation}

While the window constrains allow to model a priori knowledge,
the $t$ {\sc Ilp}s~\eqref{eq:windowilp} are uncoupled. 
As it turns out, even so, the possibilities of making efficient use of such knowledge
are rather limited.

\begin{thm} \label{thm:comb2}
{\sc Tomography under Window Constraints} is $\NP$-hard 
even if all instances are restricted to  $t=1$ and 
if the window constraints are alternatively of one of the following forms:
\begin{enumerate}
\item $\DS \sum_{j\in W^{(1)}_i}\xi_j^{(1)} = k^{(1)}_i$ with $k_i^{(1)}\in \{0,2\}$;
\item $\DS \sum_{j\in W^{(1)}_i}\xi_j^{(1)} \le k^{(1)}_i$ with $k_i^{(1)}\in \{0,2\}$.
\end{enumerate}
\end{thm}

\begin{proof}
In \cite{agsuperresolution} a problem of {\em Double Resolution under Noise},
called  \textsc{nDR}$(4)$, was introduced. There X-ray constraints for
the two coordinate directions in $\R^2$ are given, but also 
equality constraints for all disjoint $2\times 2$ windows of the form $a+[1]_0\times[1]_0$ with $a\in(2\mathbb{N}_0+1)^2$ need to be satisfied,
with the right hand sides $k_i^{(1)}\in  [4]_0.$
Further, for some of the windows, errors of up to $4$ points are permitted, 
which, of course, makes the corresponding constraints void. In~\cite{agsuperresolution} it is shown that \textsc{nDR}$(4)$ is $\mathbb{N}\mathbb{P}$-hard. The $\mathbb{N}\mathbb{P}$-hardness proof involves, apart from void constraints, only constraints of the form~(i). As this restricted \textsc{nDR}$(4)$  problem is a special case of {\sc Tomography under Window Constraints} with all constraints of the form~(i), this proves the first assertion. 

In~\cite{agwindowconstraints} a problem of \emph{reconstruction under block constraints}, called \textsc{Rec}$(2,2,0),$ was introduced, where again X-ray constraints for
the two coordinate directions in $\R^2$ are given. But, differently from \textsc{nDR}$(4)$, only $\leq$-constraints are present. They are defined for all disjoint $2\times 2$ windows of the form $a+[1]_0\times[1]_0$ with $a\in(2\mathbb{N}_0+1)^2,$ and the corresponding right hand sides are~$k_i^{(1)}\in  \{0,2\}.$  As this is an $\mathbb{N}\mathbb{P}$-hard problem and a special case of {\sc Tomography under Window Constraints} with all constraints of the form~(ii), this implies the second assertion. 
\end{proof} 

There are, however, special classes of instances of {\sc Tomography under Window Constraints} is polynomial-time solvable. We give two examples. 

Our first example deals with disjoint horizontal and vertical windows of width~1. In the tracking context, such windows $W^{(\tau)}_j$ can be viewed as to restrict the potential position of the $j$-th particle at time~$\tau.$ For each particle the positional uncertainty at time $\tau$ is allowed to extend only in either horizontal or vertical direction. These are, of course, rather special assumptions. They can be realistic, for instance, in cases where external forces generate a particle displacement field that contains only vectors orthogonal to one of the detector planes. 

\begin{thm} \label{thm:orth}
{\sc Tomography under Window Constraints} is polynomial-time solvable if the instances are restricted to satisfy the following two properties:
\begin{enumerate}
\item For every $\tau\in[t]$ and $i\in[l_\tau],$  \[W_i^{(\tau)}\subseteq \mathbb{Z}\times \{p_{i,\tau}\} \quad \textnormal{or} \quad W_i^{(\tau)}\subseteq \{p_{i,\tau}\}\times\mathbb{Z}, \qquad \textnormal{for some } p_{i,\tau} \in G^{(\tau)};\]
\item For every $\tau\in[t]$ and $i\neq j \in [l_\tau]$ it holds that $W_i^{(\tau)}\cap W_j^{(\tau)}=\emptyset.$
\end{enumerate} 
\end{thm}
\begin{proof}
Consider a fixed $\tau\in [t],$ and let again $A^{(\tau)}x^{(\tau)}=b^{(\tau)},$ $x^{(\tau)}\in\{0,1\}^{|G^{(\tau)}|}$ denote the corresponding X-ray system for the directions $S_1,S_2.$ Further, let $H^{(\tau)}x^{(\tau)}\sim^{(\tau)} k^{(\tau)}$ with a binary matrix $H^{(\tau)}$ and right-hand side vector $k^{(\tau)}$ encode the window constraints. It is well known that 
\begin{equation}\label{eq:tum0}
A^{(\tau)} \textnormal{ is the node-edge incidence matrix of a (simple) bipartite graph,}
\end{equation} simply consider the graph with vertex bipartition $V=\{v_i:(i,j)\in G^{(\tau)}\}\cup\{w_j:(i,j)\in G^{(\tau)}\}$ and edge set $E=\{(v_i,w_j):(i,j)\in G^{(\tau)}\}.$

Consider now a row vector $h^T=(h_1,\dots,h_{|G^{(\tau)}|})$ of $H^{(\tau)}.$ The support $\textnormal{supp}(h^T)=\{j:h_j\neq0\}$ of $h^T$ corresponds to the set of indices $j$ of the $p_j\in W_i^{(\tau)}$ for some $i\in l_\tau.$ These $p_j\in W_i^{(\tau)}\subseteq G^{(\tau)}$ lie, by property~(i), on a line parallel to $S_1$ or $S_2,$ hence there exists a row $a^T$ of $A^{(\tau)}$ with 
\begin{equation} \label{eq:tum1}
\textnormal{supp}(h^T)\subseteq\textnormal{supp}(a^T).
\end{equation}

Further, by property~(ii) we  have 
\begin{equation}\label{eq:tum2}
\textnormal{supp}(h'^T)\cap\textnormal{supp}(\hat{h}^T)=\emptyset
\end{equation}
for any two different row vectors $h'^T$ and $\hat{h}^T$ of $H^{(\tau)}.$

In Lemma~3.2 of~\cite{agwindowconstraints} it was shown that any matrix \[\left(\begin{array}{l}A^{(\tau)}\\H^{(\tau)}\end{array}\right)\] with $A^{(\tau)}$ satisfying~\eqref{eq:tum0} and the binary matrix~$H^{(\tau)}$ satisfying~\eqref{eq:tum1} and~\eqref{eq:tum2} is totally unimodular.
Hence, the 0/1-solutions of the linear program
\[
\begin{array}{rcl}
A^{(\tau)}x^{(\tau)} & =^{\textcolor{white}{(\tau)}} & b^{(\tau)},\\
H^{(\tau)}x^{(\tau)} &\sim^{(\tau)} & k^{(\tau)},\\
       x^{(\tau)} & \le^{\textcolor{white}{(\tau)}} &\1,\\
       x^{(\tau)} & \ge^{\textcolor{white}{(\tau)}} & 0,\\
\end{array}
\label{eq:LPtum}
\] can be determined in polynomial time (see, e.g.,~\cite[Thm.~16.2]{schrijver-86}).
\end{proof}

Our second example deals with instances that arise particularly in the context of superresolution imaging (see~\cite{agsuperresolution}). The setting is similar as in \textsc{nDR}$(4)$ and~\textsc{Rec}$(2,2,0),$ which have been considered in the proof of Theorem~\ref{thm:comb2}.

\begin{thm} \label{thm:comb5}
{\sc Tomography under Window Constraints} is polynomial-time solvable if the instances are restricted to those which have the following properties:
\begin{enumerate}
\item $G^{(1)}=\cdots=G^{(t)}=[2q]^2,$ for some $q\in\mathbb{N},$ and 
\item for each $\tau\in[t]$ and $(i,j)\in[q]^2$ there is a window constraint \[\DS \sum_{p\in (2i-1,2j-1)+[1]_0^2}\xi_p^{(\tau)} = k^{(\tau)}_{i,j} \textnormal{ with }k_{i,j}^{(\tau)}\in [4]_0.\]
\end{enumerate} 
\end{thm}
This theorem is a restatement of Theorem~2.1 from \cite{agsuperresolution}. 

Additional results for combinatorial models can be found in~\cite{agsuperresolution,agwindowconstraints}.

\section{Conclusion}\label{sect:conclusion}
We have shown that even in fairly `simple' cases, the basic problems of dynamic discrete tomography are algorithmically hard.
Of course, the $\NP$-hardness results are worst-case statements.  Therefore many real-world instances may turn out to be tractable in practice. Hence we believe that it is worthwhile to test, optimize, and compare the given algorithms on data sets obtained from experimental measurements. 
A reassuring  example of an efficient handling for experimental gliding arc data can be found in~\cite{glidingarc-15}.

It is clear that additional issues will need to be addressed to deal with noise in the measurement. This does, in particular, bring up the need to handle the lurking problems due to the inherent ill-posedness of the tomographic tasks; see~\cite{ag06,agt-01}. Also, one has to model the disappearing from and (re-)entering of particles into the camera range. 

On the positive side, we believe that the analysis of the present paper and the given algorithmic paradigms can be used to develop customized  methods for given experimental data that provide useful insight into the movement of particles in many challenging applications.

\end{document}